\documentclass[11pt]{article} 
\usepackage{fullpage}

\usepackage[compact]{titlesec}
\usepackage{etoolbox}
\newtoggle{cr}
\togglefalse{cr}

\usepackage[font=footnotesize]{caption}
\usepackage{caption}
\usepackage{nicefrac}
\usepackage{comment}
\usepackage{amsmath}
\usepackage{amsthm}
\usepackage{amssymb}
\usepackage{url}
\usepackage[final]{showkeys}
\usepackage{cleveref}
\usepackage{tikz}
\usepackage{lineno}
\usepackage[textsize=scriptsize,textwidth=2.1cm]{todonotes}

\newtheorem{theorem}{Theorem}
\newtheorem{lemma}{Lemma}

\def\mod#1#2#3{
  {#1}\equiv{#2}(\text{mod }{#3})
}

\newcommand{\N}{\mathbb{N}}
\newcommand{\poly}{\mbox{poly}}
\newcommand{\tO}{\tilde{O}}

\renewcommand{\circ}{ \mbox{{\Large $\cdot$}}}

\renewcommand{\bigcirc}{\bigodot}

\date{}
\author{
	Arturs Backurs\footnote{\texttt{backurs@mit.edu}}\\ MIT 
	\and Piotr Indyk\footnote{\texttt{indyk@mit.edu}}\\ MIT
}

\begin{document}
	\begin{titlepage}
\title{Which Regular Expression Patterns are Hard to Match?}
\clearpage\maketitle
\thispagestyle{empty}
\begin{abstract}
Regular expressions constitute a fundamental notion in formal language theory and are frequently used in computer science to define search patterns. In particular, regular expression matching  and membership testing are widely used computational primitives, employed in many programming languages and text processing utilities. A classic algorithm for these problems constructs and simulates a non-deterministic finite automaton corresponding to the expression, resulting in an $O(m n)$ running time (where $m$ is the length of the pattern and $n$ is the length of the text). This running time can be improved slightly (by a polylogarithmic factor), but no significantly faster solutions are known. At the same time, much faster algorithms exist for various special cases of regular expressions, including dictionary matching, wildcard matching, subset matching, word break problem etc.

In this paper, we show that the complexity of regular expression matching can be characterized based on its {\em depth} (when interpreted as a formula). Our results hold for expressions involving concatenation, OR, Kleene star and Kleene plus. For regular expressions of depth two (involving any combination of the above operators), we show the following dichotomy: matching and membership testing can be solved in near-linear time, except for ``concatenations of stars'', which  cannot be solved in strongly sub-quadratic time assuming the Strong Exponential Time Hypothesis (SETH). For regular expressions of depth three the picture is more complex. Nevertheless, we show that all problems can either be solved in strongly sub-quadratic time, or   cannot be solved in strongly sub-quadratic time assuming SETH. 

An intriguing special case of membership testing involves regular expressions of the form  ``a star of an OR of concatenations'', e.g., $[a|ab|bc]^*$. This corresponds to the so-called {\em word break} problem, for which a dynamic programming algorithm with a runtime of (roughly) $O(n \sqrt{m})$ is known. We show that the latter bound is not tight and improve the runtime to $O(n m^{0.44 \ldots})$.
\end{abstract}
\end{titlepage}

	\section{Introduction}

A regular expression (regexp) is a formula that describes a set of words over some alphabet $\Sigma$. It consists of individual symbols from $\Sigma$, as well as operators such as {\em OR} ``$|$" (an alternative between several pattern arguments), {\em Kleene star} ``${*}$" (which allows 0 or more repetitions of the pattern argument),  {\em Kleene plus} ``$+$" (which allows 1 or more repetitions of the pattern argument),  
wildcard ``$.$" (which matches an arbitrary symbol), 
etc.
%\footnote{Multiple notation systems are in use. In this paper we use the notation as defined in the Wikipedia page devoted to this topic~\cite{wiki}.}
For example, $[a|b]^+$ describes any sequence of symbols $a$ and $b$ of length at least $1$. See Preliminaries for the formal definition. 

In addition to being a fundamental notion in formal language theory, regular expressions are widely used in computer science to define search patterns. 
Formally, given a regular expression (pattern) $p$ of size $m$ and a sequence of symbols (text)  $t$ of length $n$, the goal of regular expression matching is to check whether a substring of $t$ can be derived from $p$.
A closely related problem is that of membership testing where the goal is to check whether the text $t$ itself can be derived from $p$.
Regular expression matching and membership testing are  widely used computational primitives, employed in several  programming languages and text processing utilities such as Perl, Python, JavaScript, Ruby, AWK, Tcl and Google RE2. Apart from text processing and programming languages, regular expressions are used in computer networks~\cite{kumar2006algorithms}, databases and data mining~\cite{garofalakis1999spirit}, computational biology~\cite{navarro2003fast}, human-computer interaction~\cite{kin2012proton} etc.

A classic algorithm for both problems constructs and simulates a non-deterministic finite automaton corresponding to the expression, resulting in the ``rectangular" $O(m n)$ running time. A sequence of improvements, first by Myers~\cite{myers1992four} and then by \cite{bille2009faster}, led to an algorithm  that achieves roughly $O(m n/\log^{1.5} n)$ running time. The latter result constitutes the fastest algorithm for this problem known to date, despite an extensive amount of research devoted to this topic. The existence of faster algorithms is a well-known open problem~(\hspace{1sp}\cite{galil1985open}, Problem 4).

However, significantly faster algorithms are known for various well-studied special cases of regular expressions. For example:
\begin{enumerate}
\item If the pattern is a concatenation of symbols (i.e., we search for a specific sequence of symbols in the text), the pattern matching problem corresponds to the ``standard'' string matching problem and can be solved in linear time, e.g., using the Knuth-Morris-Pratt algorithm~\cite{knuth1977fast}.
\item If the pattern is of the form $p_1 | p_2 | \ldots |p_k$ where $p_i$ are sequences of symbols, then the pattern matching problem corresponds to the {\em dictionary matching} problem that can be solved in linear time using the Aho-Corasick algorithm~\cite{aho1975efficient}.
\item If the pattern is a concatenation of symbols and single character wildcards ``$.$" matching any symbol, the pattern matching problem is known as the {\em wildcard matching} and can solved in (deterministic) $O(n \log m)$ time using convolutions~\cite{fischer1974string,indyk1998faster,kalai2002efficient,cole2002verifying}.
\item More generally, if the pattern is a concatenation of {\em single} character ORs of the form $s_1|s_2|\ldots |s_k$  for $s_1, \ldots, s_k \in \Sigma$ (e.g., $[a|b][a|c][b|c]$), the pattern matching problem is known as {\em superset matching} and can be solved in (deterministic) $O(n \log^2 m)$ time~\cite{cole1997tree,cole2002verifying}.
\item Finally, if the goal is to test whether a text $t$ can be derived from a pattern $p$ of the form $(p_1 | p_2 | \ldots |p_k)^+$ where $p_i$ are sequences of symbols, the problem is known as the  {\em word break} problem. It is a popular interview question~\cite{NC,LC}, and the known solutions can be implemented to run in   $\sqrt{m} n \log^{O(1)} n$ time.
%see \iftoggle{cr}{the full version of the paper}{Section~\ref{wb}} for the analysis).
\end{enumerate}

The first two examples were already mentioned in~\cite{galil1985open} as a possible reason why a faster algorithm for the general problem might be possible. Despite the existence of such examples, any super-polylogarithmic improvements to the algorithms of~\cite{myers1992four,bille2009faster} in the general case remain elusive. Furthermore, we are not aware of any systematic classification of regular expressions into ``easy'' and ``hard'' cases for the pattern matching and membership testing problems. The goal of this paper is to address this gap. 
\ \\
\paragraph{Results} Our main result is a classification of the computational complexity of regular expression matching and membership checking for patterns that involve operators ``$|$", ``$+$", ``${*}$" and concatenation ``${\circ}$", based on the pattern {\em depth}. Our classification enables us to distinguish between the cases that are solvable in sub-quadratic time (including the five problems listed above) and the cases that do not have strongly sub-quadratic time algorithms under natural complexity theoretic assumptions, such as Strong Exponential Time Hypothesis \cite{impagliazzo2001complexity}  and Orthogonal Vectors conjecture~\cite{williams2005new,followup2}. Our results therefore demonstrate a non-trivial dichotomy for the complexity of these problems. 

To formulate our results, we consider pattern formulas that are {\em homogeneous}, i.e., in which the operators at the same level of the formula are equal (note that the five aforementioned problems involve patterns that satisfy this condition). We say that  a homogeneous   formula of depth $k$ has {\em type} $o_1 o_2 \ldots o_k$ if for all levels $i$ the operators at level $i$ are equal to $o_i$ (note that, in addition to the operators, a level might also contain leaves, i.e. symbols;
for example the expression  $[a|b]a[b|c]$ is a depth-$2$ formula of type ``${\circ} |$"). We will assume that no two consecutive operators in the type descriptor are equal, as otherwise they can be collapsed into one operator. 

Our results are described in Table~\ref{t:d2} (for depth-2 expressions) and Table~\ref{t:d3} (for depth-3 expressions). The main findings can be summarized as follows:
\begin{enumerate}
\item Almost all pattern matching and membership problems involving depth-2 expressions can be solved in near-linear time. The lone exception involve patterns of type ``${\circ} {*}$'', for which we show that matching and membership problems cannot be solved in time $O((mn)^{1-\delta})$ for any constant $\delta>0$ and $m\leq n$ assuming the Strong Exponential Time Hypothesis (SETH)~\cite{impagliazzo2001complexity}.   Interestingly, we show that pattern matching with a very similar depth-2 type, namely ``${\circ} +$'', can be solved in $O(n \log^2 m)$ time. 
\item Pattern matching problems with depth-2 expressions contain a ``high density'' of interesting algorithmic problems, with non-trivial algorithms existing for types ``${\circ} +$'' (this paper), ``${\circ} |$''~\cite{cole2002verifying} , ``$| {\circ}$'' \cite{aho1975efficient} and ``$+  {\circ}$'' (essentially solved in \cite{knuth1977fast}, since $+$ can be dropped). In contrast, membership problems with depth-2 expressions have a very restrictive structure that makes them mostly trivially solvable in linear time, with the aforementioned exception for the ``${\circ} {*}$'' type
\item Pattern matching problems with depth-3 expressions have a more diversified structure. All types starting with  ${\circ}$ are SETH-hard; all types starting with $|$ are either-SETH hard (if followed by ${\circ}$) or easily solvable in linear time; all types starting from ${*}$ are trivially solvable in linear time (since ${*}$ allows zero repetitions); all types starting from $+$ inherit their complexity from the last two operators in the type description (since $+$ allows exactly one repetition). 
\item Finally, membership checking with depth-3 expressions presents the most complex picture.  As before, all types starting with  ${\circ}$ are SETH-hard. On the other hand, types starting with $|$ (except for $| {\circ} {*}$) have linear time algorithms, with difficulty levels that range from trivial observations to undergraduate-level exercises\footnote{Incidentally, the second author used some of the problems from Table~\ref{t:d3} in an undergraduate algorithms course.}. Types starting with ${*}$ or $+$ include ``${*} | {\circ}$'' and ``$+ | {\circ}$'', which correspond to the aforementioned word break problem~\cite{NC,LC}. This is the only problem in the table whose (conditional) complexity is not determined up to logarithmic factors. However, we show that the running time of the standard dynamic-programming based algorithm can be improved, from roughly $n m^{0.5}$ to roughly $n m^{0.5 \ - \ \nicefrac{1}{18}}$. 
\end{enumerate}

\begin{table*}
\normalsize
\hspace*{-0.8cm} 
\begin{tabular}{| l | l | l | l |}
\hline
Type & Example & Pattern matching & Membership\\
\hline
${\circ}  +$ & $a^+ a b^+$ & $O(n \log^2 m)$ (\iftoggle{cr}{full version}{Section~\ref{cp}}) &  $O(n+m)$ (immediate) \\
${\circ}  {*}$ & $a^* a b^*$ & $\Omega((mn)^{1-\alpha})$ (Section~\ref{cs}) &  $\Omega((mn)^{1-\alpha})$ (\iftoggle{cr}{full version}{Section~\ref{mcs}}) \\
${\circ}  |$& $[a|b] [b|c]$ & $O(n \log^2 m)$ ({\bf superset matching} \cite{cole2002verifying})& $O(n+m)$ (immediate)\\
\hline
$|  {\circ}$& $ab | c$ & $O(n+m)$ ({\bf dictionary matching} \cite{aho1975efficient})& $O(n+m)$ (immediate)\\
$|  {*}$ &   $a^* |a | b^*$ &  $O(n+m)$ (immediate)  &  $O(n+m)$ (immediate) \\
$|  +$ & $a^+ | a | b^+$ & $O(n+m)$ (immediate - reducible to ``$|$'') & $O(n+m)$ (immediate) \\
\hline
${*}  {\circ}$ & $[ab]^*$ &  $O(n+m)$ (immediate)   & $O(n+m)$ (immediate) \\
${*}  +$ &   $[a^+]^*$ &  $O(n+m)$ (immediate)  & $O(n+m)$ (immediate - reducible to ``$+$'') \\
${*}  |$ & $[a|b]^*$ &  $O(n+m)$ (immediate)  & $O(n+m)$ (immediate) \\
\hline
$+  {\circ}$ & $[ab]^+$ & $O(n+m)$ ({\bf string matching}  \cite{knuth1977fast})& $O(n+m)$ (immediate - equivalent to ``${*} {\circ}$'')  \\
$+  |$ &  $[a|b]^+$ & $O(n+m)$ (immediate - reducible to ``$|$'') & $O(n+m)$ (immediate - equivalent to ``${*} |$'') \\
$+  {*}$ & $[a^*]^+$ & $O(n+m)$ (immediate) &  $O(n+m)$ (immediate - reducible to ``$+$'') \\
\hline
\end{tabular}
\caption{Classification of the complexity of the pattern matching and the membership test problems for depth-2 expressions.  All lower bounds assume SETH and $m \le n$.
Some upper bounds use randomization (notably hashing). 
For an explanation of reducibility see ``Our techniques''.}
\label{t:d2}
\end{table*}

\begin{table*}
\normalsize
\iftoggle{cr}{}{\hspace*{-0.8cm}}
\begin{tabular}{| l | l | l | l |}
\hline
Type & Example & Pattern matching & Membership\\
\hline
${\circ} | {\circ}$ & $[a|bb][ba|b]$ &  $\Omega((mn)^{1-\alpha})$ (\iftoggle{cr}{full version}{Section~\ref{coc}})    &   $\Omega((mn)^{1-\alpha})$ (\iftoggle{cr}{full version}{Section~\ref{mcoc}})  \\
${\circ} | {*}$ & $[a^* | b^* ][c^*|b]$ & $\Omega((mn)^{1-\alpha})$  (from ``${\circ} {*}$'')     &     $\Omega((mn)^{1-\alpha})$  (from ``${\circ} {*}$'')  \\
${\circ} | +$ & $[a^+ | b^+ ][c^+|b]$ &  $\Omega((mn)^{1-\alpha})$ (\iftoggle{cr}{full version}{Section~\ref{cop}})      & $\Omega((mn)^{1-\alpha})$ (\iftoggle{cr}{full version}{Section~\ref{mcop}})   \\
${\circ} {+} {\circ}$ &  $[ab]^+[bca]^+$ & $\Omega((mn)^{1-\alpha})$ (\iftoggle{cr}{full version}{Section~\ref{cpc}})     &   $\Omega((mn)^{1-\alpha})$ (\iftoggle{cr}{full version}{Section~\ref{mcpc}})   \\
${\circ} {+} |$ & $[a|b]^+[a|c|d]^+$ &   $\Omega((mn)^{1-\alpha})$   (\iftoggle{cr}{full version}{Section~\ref{cpo}})    &   $\Omega((mn)^{1-\alpha})$   (\iftoggle{cr}{full version}{Section~\ref{mcpo}})    \\
${\circ} {+} {*}$ & $[a][a^+]^* [b^+]$ &  $\Omega((mn)^{1-\alpha})$  (from ``${\circ} {*}$'')     &  $\Omega((mn)^{1-\alpha})$  (from ``${\circ} *$'')   \\
${\circ} * {\circ}$ & $[ab]^*[bca]^*$ &  $\Omega((mn)^{1-\alpha})$ (from ``${\circ} *$'')       &  $\Omega((mn)^{1-\alpha})$  (from ``${\circ} *$'')   \\
${\circ} {*} |$ & $[a|b]^*[a|b|c]^*$ &  $\Omega((mn)^{1-\alpha})$ (from ``${\circ} *$'')     &  $\Omega((mn)^{1-\alpha})$  (from ``${\circ} *$'')   \\
${\circ} * {+}$ & $[a^*]b[b^+]^*$ & $\Omega((mn)^{1-\alpha})$  (from ``${\circ} *$'')       &  $\Omega((mn)^{1-\alpha})$  (from ``${\circ} *$'')   \\
\hline
$| {\circ} |$ & \iffalse $[a|b] | [a|b|c]$ \fi $[(a|b)(b|c)]  | [(a|c)b]$ &  $\Omega((mn)^{1-\alpha})$ (Section~\ref{oco})      &   $O(n+m)$ (immediate) \\
$| {\circ} *$ & $[a^*b^*] | [b^*c^*]$ &   $\Omega((mn)^{1-\alpha})$  (from ``${\circ} *$'')    &   $\Omega((mn)^{1-\alpha})$  (from ``${\circ} *$'')  \\
$| {\circ} {+}$ & $[a^+b^+] |  [b^+c^+]$ & $\Omega((mn)^{1-\alpha})$  (Section~\ref{ocp})     & $O(n+m)$ (\iftoggle{cr}{full version}{Section~\ref{run}})  \\
$| {*} {\circ}$ & $[abc]^*  |  [bc]^*$ &  $O(n+m)$ (immediate - reducible to ``$|$'')      & $O(n+m)$   (\iftoggle{cr}{full version}{Section~\ref{period}})   \\
$| {*} |$ & $[a|b|c]^*  |  [b|c]^*$ & $O(n+m)$ (immediate - reducible to ``$|$'')       &   $O(n+m)$ (immediate) \\
$| {*} {+}$ & $[a^+]^*  |  [b^+]^*$ &  $O(n+m)$ (immediate - reducible to ``$|$'')       & $O(n+m)$ (immediate)    \\
$| {+} {\circ}$ & $[abc]^+  |  [bc]^+$ &  $O(n+m)$  (reducible to ``$| {\circ}$'')       & $O(n+m)$   (\iftoggle{cr}{full version}{Section~\ref{period}})  \\
$| {+} |$ & $[a|b|c]^+  |  [b|c]^+$ &  $O(n+m)$   (immediate - reducible to ``$|$'')       &  $O(n+m)$ (same as ``$| {*} |$'')   \\
$| {+} {*}$ & $[a^*]^+  |  [b^*]^+$ & $O(n+m)$  (immediate - reducible to ``$|$'')       &   $O(n+m)$ (immediate)  \\
\hline
${*} {\circ} |$ & $[[a|b][b|c]]^*$  &  $O(n+m)$ (immediate)  &  $O(n+m)$ (immediate)  \\
${*} {\circ} {*}$ & $[a^*b^*c^*]^*$ & $O(n+m)$ (immediate)   &  $\Omega((mn)^{1-\alpha})$  (from ``${\circ} {*}$'')  \\
${*} {\circ} {+}$ & $[a^+b^+c^+]^*$ & $O(n+m)$ (immediate)   &   $O(n+m)$ (\iftoggle{cr}{full version}{Section~\ref{run}})  \\
${*} | {\circ}$ & $[a|ab|bc]^*$  & $O(n+m)$ (immediate)   & $O(n m^{0.44 \ldots})$ ({\bf word break} - \iftoggle{cr}{full version}{Section~\ref{wb}})  \\
${*} | {*}$ & $[a^*|b^*|c^*]^*$  &  $O(n+m)$ (immediate)  &    $O(n+m)$ (immediate)  \\
${*} | {+}$ & $[a^+|b^+|c^+]^*$  & $O(n+m)$ (immediate)   &    $O(n+m)$ (immediate)  \\
${*} {+} {\circ}$ & $[[abcd]^+]^*$  & $O(n+m)$ (immediate)   &   $O(n+m)$ (immediate)  \\
${*} {+} |$ & $[[a|b|c|d]^+]^*$  & $O(n+m)$ (immediate)   &    $O(n+m)$ (immediate) \\
${*} {+} {*}$ & $[[a^*]^+]^*$  &  $O(n+m)$ (immediate)  &   $O(n+m)$ (immediate)  \\
\hline
${+} {\circ} |$ & $[[a|b][b|c]]^+$  & $O(n \log^2 m)$  (reducible to ``${\circ} |$'')  & $O(n+m)$ (same as ``${*} {\circ} |$'')   \\
${+} {\circ} {*}$ & $[a^*b^*c^*]^+$ & $\Omega((mn)^{1-\alpha})$  (from ``${\circ} {*}$'')      &   $\Omega((mn)^{1-\alpha})$ (same as ``${*} {\circ} {*}$'') \\
${+} {\circ} {+}$ & $[a^+b^+c^+]^+$  &  $O(n \log^2 m)$  (reducible to ``${\circ} + $'') & $O(n+m)$ (same as ``${*} {{\circ}} +$'')    \\
${+} | {\circ}$ & $[a|ab|bc]^+$  &  $O(n+m)$ (reducible to ``$ | {\circ}$'') &   $O(n m^{0.44 \ldots})$ ({\bf word break} - \iftoggle{cr}{full version}{Section~\ref{wb}})  \\
${+} | {*}$ & $[a^*|b^*|c^*]^+$  &  $O(n+m)$  (reducible to ``$| {*}$'')  & $O(n+m)$ (same as ``${*} | {*}$'')\\
${+} | {+}$ & $[a^+|b^+|c^+]^+$  &  $O(n+m)$  (reducible to ``$ | +$'') & $O(n+m)$  (same as  ``${*} | {+}$'') \\
${+} {*} {\circ}$ & $[[abcd]^*]^+$ &  $O(n+m)$    (reducible to ``${*} {\circ}$'') & $O(n+m)$  (same as ``${*} {\circ}$'')   \\
${+} {*} |$ & $[[a|b|c|d]^*]^+$ & $O(n+m)$   (reducible to ``${*} |$'') & $O(n+m)$  (same as ``${*} |$'')   \\
${+} {*} {+}$ & $[[a^+]^*]^+$ & $O(n+m)$    (reducible to ``${*} {+}$'')   & $O(n+m)$   (same as ``${*} {+}$'')  \\
\hline
\end{tabular}
\caption{Classification of the complexity of the pattern matching and the membership test problems for depth-3 expressions. See Figure~\ref{figure_pattern_tree} for the visualization of the table.}
\label{t:d3}
\end{table*}

\paragraph{Our techniques} 
%Consider first the case of the $\{ +, |, {\circ} \}$ operator set. 
Our upper bounds for depth-2 expressions follow either from known near-linear time algorithms for specific variants of regular expressions, or relatively simple constructions of such algorithms. In particular, we observe that type ``${\circ} |$" expressions (concatenations of ORs)  correspond to superset matching, type ``$|{\circ}$" expressions (OR of sequences) correspond to dictionary matching and type ``$+{\circ}$'' reduces to ``${\circ}$'' and thus corresponds to the standard pattern matching problem. Furthermore, we give a near-linear time algorithm for pattern matching with type ``${\circ} +$" expressions, where patterns are concatenations of expressions of the form $s^{\ge k}$ or $s^{k}$, where $s^{\ge k}$ denotes a sequence of symbols $s$ repeated at least $k \ge 1$ times.\footnote{For example, the expression    $aa^+bc^+$ generates all words of the form $a^{\ge 2}b^1c^{\ge 1}$.} We show that this problem can be solved in near-linear time by reducing it to one instance of subset matching and one instance of wildcard matching. All other problems can be solved in linear time, with the exception of type ``${\circ} {*}$". The latter expressions correspond to patterns obtained by concatenating  patterns of the form $s^{\ge k}$ and  $s^{k}$. Unlike in the ``${\circ} +$" case, however, here we cannot assume that $k \ge 1$, since each symbol could be repeated zero times. We show that this simple change makes the problem SETH-hard. 
 This is accomplished by a reduction from an intermediate problem, namely the (unbalanced version of the) Orthogonal Vectors Problem (OVP)~\cite{williams2005new,followup2}.
 The problem is defined as follows: given  two sets $A,B\subseteq \{0,1\}^d$ such that $|A|=M$ and $|B|=N$, determine whether there exists $x \in A$ and $y \in B$ such that the dot product  $x\cdot y =\sum_{j=1}^d x_j y_j$  is equal to $0$.\footnote{ The reduction is somewhat complex, so we will not outline it here. However, we give an overview of other reductions from OVP in the next few paragraphs.}

\iffalse
Furthermore, type ``$| +$", type ``$+ |$" and type ``$+ {\circ}$" expressions  can be transformed into depth-1 expressions (type ``$|$", type ``$|$" and type ``${\circ}$", respectively), and therefore are easy to search for. The only remaining case to consider is 
%We show that this problem in linear time can be reduced to an instance of subset matching and an instance of pattern matching with wildcards, which yields a near-linear time algorithm for this problem. 
\fi

Our results for depth-3 expression pattern matching are multi-fold. First, all types starting from ${*}$ are trivially solvable in linear time, since ${*}$ allows zero repetitions. Second, all types starting from $+$ inherit their complexity from the last two operators in the type description, since $+$ allows exactly one repetition. Third, all types starting from $| {*}$ or $| +$ have simple linear time solutions. 

The remaining cases lead to SETH-hard problems. For six types this follows immediately from the analogous result for type ``${\circ} {*}$".
For the six remaining types the hardness is shown via individual reductions from OVP. 
For some types, such a reduction is immediate. For example, for type ``$| {\circ} |$" expressions (ORs of concatenations of ORs), we form the text by concatenating all vectors in $B$ (separated by some special symbol), and we form the pattern by taking an OR of the vectors in $A$, modified by replacing $0$ with $[0|1]$ and $1$ with $0$. A similar approach works for type ``$| {\circ} +$" expressions. 

The remaining four types are grouped into two classes:  ``${\circ} {+} {\circ}$" is grouped with ``${\circ} | {\circ}$" and ``${\circ} {+} |$" is grouped with ``${\circ} | +$".
For each group, we first show hardness of the first type in the group (i.e., of ``${\circ} {+} {\circ}$" and ``${\circ} {+} |$", respectively). We then show that the second type in each group is hard by making changes to the hardness proof for the first type. 

\begin{figure*}
	\centering
	
	\begin{tikzpicture} [solid]
		%\draw[fill=black] (0,0) circle (0.2mm);
		%\draw[fill=black] (3,3) circle (0.2mm);
		%\draw[fill=black] (3,1) circle (0.2mm);
		%\draw[fill=black] (3,-1) circle (0.2mm);
		%\draw[fill=black] (3,-3) circle (0.2mm);
		\node at (0,3) {Pattern matching:};
		
		\draw (0,0) .. controls (1,1) and (2,2.7) .. (3,3);
		\node[anchor=south east] at (1.5,1.5) {$\circ$};
		\node[anchor=west] at (3,3) {hard};
		
		\draw (0,0) .. controls (1,-1) and (2,-2.7) .. (3,-3);
		\node[anchor=south east] at (1.5,0.45) {$|$};
		
		\draw (0,0) .. controls (1,0.2) and (2,0.9) .. (3,1);
		\node[anchor=north east] at (1.5,-0.45) {$*$};
		
		\draw (0,0) .. controls (1,-0.2) and (2,-0.9) .. (3,-1);
		\node[anchor=north east] at (1.5,-1.5) {$+$};
		\node[anchor=west] at (3,-1) {easy};
		
		\draw (3,1) .. controls (3.33,1.2) and (3.66,1.54) .. (4,1.6);
		\node[anchor=south east] at (3.6,1.25) {$\circ$};
		\node[anchor=west] at (4,1.6) {hard};
		
		\draw (3,1) .. controls (3.33,0.8) and (3.66,0.46) .. (4,0.4);
		\node[anchor=north east] at (3.7,0.7) {${*},{+}$}; 
		\node[anchor=west] at (4,0.4) {easy};
		
		\draw (3,-4+1) .. controls (3.33,-4+1.2) and (3.66,-4+1.54) .. (4,-4+1.6);
		\node[anchor=south east] at (3.6,-4+1.25) {$\circ$};
		
		\draw (4,-2.4) .. controls (4.2,-2.3) and (4.6,-2.05) .. (4.8,-2);
		\node[anchor=south east] at (4.5,-2.2) {\footnotesize $*$};
		\node[anchor=west] at (4.8,-2) {hard};
		
		\draw (4,-2.4) .. controls (4.2,-2.5) and (4.6,-2.75) .. (4.8,-2.8);
		\node[anchor=north east] at (4.6,-2.5) {\footnotesize ${|},{+}$}; 
		\node[anchor=west] at (4.8,-2.8) {easy};
		
		\draw (3,-4+1) .. controls (3.33,-4+0.8) and (3.66,-4+0.46) .. (4,-4+0.4);
		\node[anchor=north east] at (3.7,-4+0.7) {${|},{*}$}; 
		\node[anchor=west] at (4,-4+0.4) {easy};
		
		\begin{scope}[shift={(8,0)}]
			\node at (0,3) {Membership:};
			
			\draw (0,0) .. controls (1,1) and (2,2.7) .. (3,3);
			\node[anchor=south east] at (1.5,1.5) {$\circ$};
			\node[anchor=west] at (3,3) {hard};
			
			\draw (0,0) .. controls (1,-1) and (2,-2.7) .. (3,-3);
			\node[anchor=south east] at (1.5,0.45) {$|$};
			
			\draw (0,0) .. controls (1,0.2) and (2,0.9) .. (3,1);
			\node[anchor=north east] at (1.5,-0.45) {$*$};
			
			\draw (0,0) .. controls (1,-0.2) and (2,-0.9) .. (3,-1);
			\node[anchor=north east] at (1.5,-1.5) {$+$};
			
			\draw (3,1) .. controls (3.33,1.2) and (3.66,1.54) .. (4,1.6);
			\node[anchor=south east] at (3.6,1.25) {$\circ$};
			
			\draw (4,4-2.4) .. controls (4.2,4-2.3) and (4.6,4-2.05) .. (4.8,4-2);
			\node[anchor=south east] at (4.5,4-2.2) {\footnotesize $*$};
			\node[anchor=west] at (4.8,4-2) {hard};
			
			\draw (4,4-2.4) .. controls (4.2,4-2.5) and (4.6,4-2.75) .. (4.8,4-2.8);
			\node[anchor=north east] at (4.6,4-2.5) {\footnotesize ${|},{+}$}; 
			\node[anchor=west] at (4.8,4-2.8) {easy};
			
			\draw (3,1) .. controls (3.33,0.8) and (3.66,0.46) .. (4,0.4);
			\node[anchor=north east] at (3.7,0.7) {${*},{+}$}; 
			\node[anchor=west] at (4,0.4) {easy};
			
			%%%
			
			\draw (3,-2+1) .. controls (3.33,-2+1.2) and (3.66,-2+1.54) .. (4,-2+1.6);
			\node[anchor=south east] at (3.6,-2+1.25) {$\circ$};
			
			\draw (4,-0.4) -- (4.8,-0.3);
			\node[anchor=south east] at (4.6,-0.4) {\footnotesize $*$};
			\node[anchor=west] at (4.8,-0.25) {\footnotesize hard};
			\draw (4,-0.4) -- (4.8,-0.5);
			\node[anchor=north east] at (4.7,-0.4) {\tiny ${|},{+}$};
			\node[anchor=west] at (4.8,-0.55) {\footnotesize easy};
			
			\node[anchor=south] at (3.7,-1.1) {\small $|$};
			\draw (3,-1) -- (4,-1);
			
			\draw (4,-0.6-0.4) -- (4.8,-0.6-0.3);
			\node[anchor=south east] at (4.6,-0.7-0.4) {\footnotesize $\circ$};
			\node[anchor=west] at (4.8,-0.6-0.25) {\footnotesize Word Break};
			\draw (4,-0.6-0.4) -- (4.8,-0.6-0.5);
			\node[anchor=north east] at (4.7,-0.6-0.4) {\tiny ${*},{+}$};
			\node[anchor=west] at (4.8,-0.6-0.55) {\footnotesize easy};

			\draw (3,-2+1) .. controls (3.33,-2+0.8) and (3.66,-2+0.46) .. (4,-2+0.4);
			\node[anchor=north east] at (3.7,-2+0.7) {$+$}; 
			\node[anchor=west] at (4,-2+0.4) {easy};
			
			%%%
			
			\draw (3,-2-2+1) .. controls (3.33,-2-2+1.2) and (3.66,-2-2+1.54) .. (4,-2-2+1.6);
			\node[anchor=south east] at (3.6,-2-2+1.25) {$\circ$};
			
			\draw (4,-2-0.4) -- (4.8,-2-0.3);
			\node[anchor=south east] at (4.6,-2-0.4) {\footnotesize $*$};
			\node[anchor=west] at (4.8,-2-0.25) {\footnotesize hard};
			\draw (4,-2-0.4) -- (4.8,-2-0.5);
			\node[anchor=north east] at (4.7,-2-0.4) {\tiny ${|},{+}$};
			\node[anchor=west] at (4.8,-2-0.55) {\footnotesize easy};
			
			\node[anchor=south] at (3.7,-2-1.1) {\small $|$};
			\draw (3,-2-1) -- (4,-2-1);
			
			\draw (4,-2-0.6-0.4) -- (4.8,-2-0.6-0.3);
			\node[anchor=south east] at (4.6,-2-0.7-0.4) {\footnotesize $\circ$};
			\node[anchor=west] at (4.8,-2-0.6-0.25) {\footnotesize Word Break};
			\draw (4,-2-0.6-0.4) -- (4.8,-2-0.6-0.5);
			\node[anchor=north east] at (4.7,-2-0.6-0.4) {\tiny ${*},{+}$};
			\node[anchor=west] at (4.8,-2-0.6-0.55) {\footnotesize easy};

			\draw (3,-2-2+1) .. controls (3.33,-2-2+0.8) and (3.66,-2-2+0.46) .. (4,-2-2+0.4);
			\node[anchor=north east] at (3.7,-2-2+0.7) {$*$}; 
			\node[anchor=west] at (4,-2-2+0.4) {easy};
		\end{scope}
	\end{tikzpicture}
	
	\caption{Tree diagrams visualizing Table~\ref{t:d3}. Depth 3 types are classified as ``easy'' (near-linear time), ``hard'' (near-quadratic time, assuming SETH), or ``Word Break'' (whose complexity is not determined). The leftmost operators in each tree correspond to the leftmost operators in type descriptions.}
	\label{figure_pattern_tree}
\end{figure*}
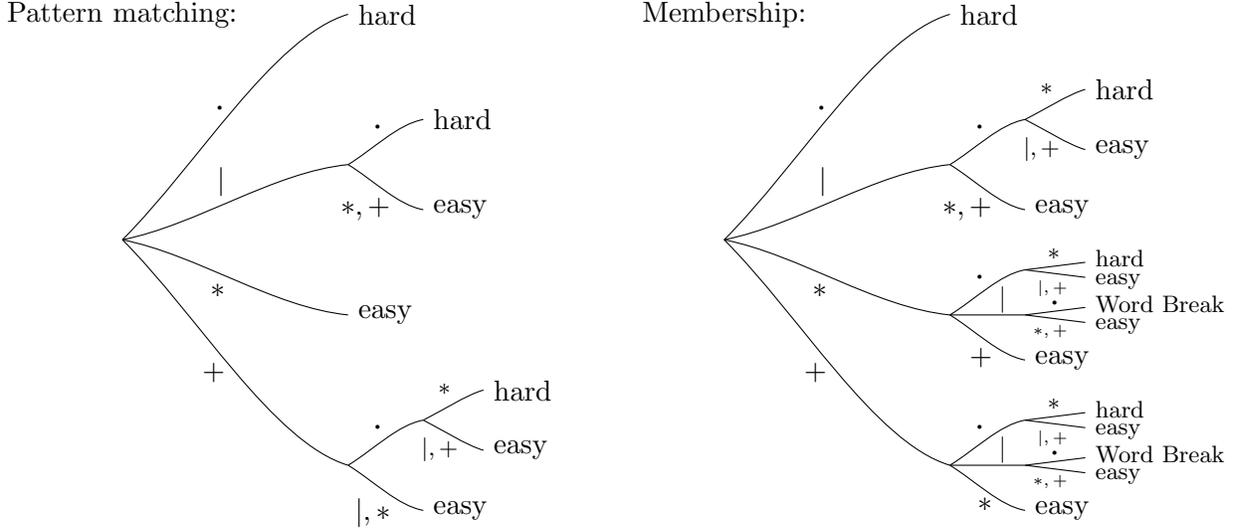

The hardness proof for ``${\circ} {+} {\circ}$"  proceeds as follows. 
We form the pattern $p$ by concatenating (appropriately separated) pattern vector gadgets for each vector in $A$, and form the text $t$ by concatenating (appropriately separated) text vector gadgets for each vector in $B$. We then show that if there is a pair of orthogonal vectors $a^i \in A, b^j \in B$ then $p$ can be matched to a substring of $t$, and vice versa. To show this,  we construct $p$ and $t$ so that any pair of gadgets (in particular the gadgets for $a^i$  and $b^j$) can be aligned.  We then show that (i) each vector gadget for a vector in $A$ can be matched with ``most" of the gadget for the corresponding $b \in B$ (ii) matching the gadgets for  {\em orthogonal} vectors $a^i$ and $b^j$ allows us to make a ``smaller step", i.e., to match the gadget for $a^i$ with a smaller part of the gadget for $b^j$, and (iii) at least one ``smaller step" is necessary to completely derive a substring of $t$ from $p$. We then conclude that there is a pair of orthogonal vectors $a^i \in A, b^j \in B$ if and only if $p$ can be matched to a substring of $t$. The hardness proof for  ``${\circ} {+} |$"  follows a similar general approach, although the technical development is different. In particular, we construct the gadgets such that the existence of orthogonal vectors makes it possible to make a ``bigger step", i.e., to derive a bigger part of $t$, and that one bigger step is necessary to complete the derivation. 

To show hardness of the second type in each group, we adapt the arguments for the first type in the group. In particular, to show hardness for type ``${\circ} | {\circ}$" , we construct $p$ and $t$ as in the reduction for type ``${\circ} {+} {\circ}$" and then transform $p$ into a type ``${\circ} | {\circ}$"  regular expression $p'$.
The transformation has the property that $p'$  is {\em less} expressive than $p$ (i.e., the language corresponding to $p$ is a \emph{superset} of the language corresponding to $p'$), but the specific substrings of the text $t$ needed for the reduction can be still derived from $p'$. The hardness proof for  ``${\circ} | +$" is obtained via a similar transformation of the hardness proof for ``${\circ} {+} |$". 

Finally, consider the membership checking problem for depth-3 expressions.  As before, all types starting with  ${\circ}$ are  shown  to be SETH-hard. The reductions are similar to those for the pattern matching problem, but in a few cases require some modifications.
On the other hand, types starting with $|$ (with the exception of $| {\circ} {*}$) have linear time algorithms. 
The algorithms are not difficult, but require the use of basic algorithmic notions, such as periodicity (for types ``$| {*} {\circ}$'' and ``$| {+} {\circ}$'') and run-length encoding (for type ``$| {\circ} +$'').
Types starting with ${*}$ are mostly solvable in linear time, with two exceptions: type ``${*} {\circ} {*}$'' inherits the hardness from ``${\circ} {*}$''\footnote{Let a regular expression $p$ and a text $t$ be a hard instance for the ``${\circ} {*}$'' membership problem. Let $x$ be a symbol that does not appear in $p$ or $t$. Then $p':=(x \circ p \circ x)^*$ and $t:=xtx$ is a hard instance for the membership problem of type ``${*} {\circ} {*}$''. Since $t'$ starts and ends with a unique symbol, we must use the argument regular expression $x \circ p \circ x$ exactly once. Thus we get ``${\circ} {*}$'' membership problem.}, while the type  ``${*} | {\circ}$''  corresponds to the aforementioned word break problem which we discuss in the next paragraph. Finally, types starting with $+$ are analogous to those starting with ${*}$.

The word break problem is the only problem in the table whose (conditional) complexity is not determined up to logarithmic factors.  There are several known solutions to this problem based on dynamic programming~\cite{NC,LC}. A careful implementation of those algorithms (using substring hashing and pruning)  leads to a runtime of $O( nm^{0.5} \log^{O(1)} n)$.
%  (see \iftoggle{cr}{the full version}{Section~\ref{wb}}). 
However, we show that this bound is not tight, and can be further improved  to roughly $O(n m^{0.5 \ - \ \nicefrac{1}{18}})$. Our new algorithm speeds up the dynamic program by using convolutions to pre-compute information that is reused multiple times during the execution of the algorithm. We note that the algorithm is randomized and has a one-sided error. \\
\iftoggle{cr}{Due to space limitations, most of the proofs have been omitted from this version of the paper.}{}

\paragraph{Related work} Our hardness results come on the heels of several recent works demonstrating quadratic hardness of sequence alignment problems assuming SETH or other conjectures. In particular, such results have been shown for Local Alignment~\cite{AWW}, Fr\'echet distance~\cite{Bring}, Edit Distance~\cite{BI} and Longest Common Subsequence~\cite{followup1,followup2}. As in our case, most of those results were achieved by a reduction from OVP, performed by concatenating appropriately constructed gadgets for the vectors in the input sets. The technical development in our paper is, however, quite different, since regular expression matching is not defined  by a sequence similarity  measure. Instead, our gadget constructions are tailored to the specific sets of operators and expression types defining the problem variants. Furthermore, we exploit the similarity between related expression types (such as ``${\circ} {+} {\circ}$" and ``${\circ} | {\circ}$") and show how to convert a hardness proof for one type into a hardness proof for the other type.

The reduction in Section~\ref{oco} has been independently discovered by Kasper Larsen and Raphael Clifford (personal communication). Conditional lower bounds (via reductions from 3SUM) for certain classes of regular expressions have been investigated in~\cite{amir2015mind}. Estimating the complexity of regular expression matching using {\em specific} algorithms  has also been  a focus of  several papers.  See e.g., ~\cite{weideman2016analyzing} and the references therein.

	\section{Preliminaries}
\ \\
\paragraph{Orthogonal Vectors problem} Our reductions use the unbalanced version of the {\em Orthogonal Vectors Problem}, defined as follows: given  two sets $A,B\subseteq \{0,1\}^d$ such that $|A|=M, |B|=N$, determine whether there exists $x \in A$ and $y \in B$ such that the dot product  $x\cdot y =\sum_{j=1}^d x_j y_j$ (taken over the reals) is equal to $0$.   
An equivalent formulation of this problem is:  given two collections of subsets of $\{1, \ldots, d\}$,  of sizes $M$ and $N$, respectively,  determine if  there is a set in the first collection that is contained in a set from the second collection. 

  It is known that, for any $M=\Theta(N^{\alpha})$ for some $\alpha \in (0,1]$ and any constant $\delta>0$,  any algorithm for OVP problem with an $O(MN)^{1-\delta}$ running time  would also yield a more efficient algorithm for SAT, violating the Strong Exponential Time Hypothesis, even in the setting when the dimension $d$ is arbitrary $d=\omega(\log N)$. 
  This was shown in~\cite{williams2005new} for the balanced case $M=N$, and extended to the unbalanced case in~\cite{followup2}. Therefore, in this paper we show that a problem is SETH-hard by reducing unbalanced OVP to it. \\
\paragraph{Subset Matching problem} In the subset matching problem, we are given a pattern string $p$ and a text string $t$ where each pattern and text location is a set of symbols drawn from some alphabet.  The pattern is said to occur at the text position $i$ if the set $p_j$ is a subset of the set $t_{i+j}$ for all $j$. The goal of the problem is find all positions where $p$ occurs in $t$. The problem can be solved in deterministic $O(N \log^2 N)$ time, where $N=\sum_i |t_i| + \sum_i |p_i|$ \cite{cole2002verifying}.\\
\paragraph{Superset Matching problem} This problem is analogous to subset matching except that we require that $p_j$ is a \emph{superset} of the set $t_{i+j}$ for all $j$. The aforementioned algorithm of~\cite{cole2002verifying} applies to this problem as well.\\
\paragraph{Wildcard Matching problem} In the wildcard matching problem, we are given a pattern string $p$ and a text string $t$ where each pattern and text location is an element from $\Sigma \cup \{.\}$, where ``$.$" is the special wildcard symbol.   The pattern is said to occur at the text position $i$ if for all j we have that (i) one of the symbols $p_j$ and  $t_{i+j}$ is equal to ``$.$", or (ii)  $p_j=t_{i+j}$. The goal of the problem is find all positions where $p$ occurs in $t$. The problem can be solved in deterministic $O(n \log n)$ time \cite{cole2002verifying}.\\
\paragraph{Regular expressions} 
Regular expressions over a symbol set $\Sigma$ and an operator set $O=\{ \circ, |, +, * \}$ are defined recursively as follows:
\begin{itemize}
\itemsep0em
\item $a$ is a regular expression, for any $a \in \Sigma$;
\item if $R$ and $S$ are regular expressions then so are $[R]|[S]$, $[R]\circ[S]$, $[R]^+$ and $[R]^*$.
\end{itemize}

For the sake of simplicity,  in the rest of this paper we will typically omit the concatenation operator $\circ$, and also omit some of the parenthesis if the expression is clear from the context. 

A regular expression $p$ determines a language $L(p)$ over $\Sigma$. Specifically, for any regular expressions $R$, $S$ and any $a \in \Sigma$, we have: $L(a)=\{a\}$; $L(R|S) =L(R) \cup L(S)$; $L(R \circ S)=\{uv: u \in L(R), v \in L(S)\}$; $L(R^+) = \cup_{i \ge 1} \circ_{j=1}^i L(R)$; and $L(R^*)=L(R^+) \cup \{ \epsilon \}$, where $\epsilon$ denotes the empty word. 

Regular expressions can be viewed as rooted labeled trees, with internal nodes labeled with operators from  $O$ and leaves labeled with symbols from $\Sigma$. Note that the number of children of an internal node is not fixed, and can range between $1$ (for $+$ and $*$) and $m$. We say that a regular expression is {\em homogeneous} if all internal node labels at the same tree level are equal. Note that this definition does not preclude expressions such as $a a^+$ where not all leaves have the same depth.
A homogeneous formula of depth $k$ has {\em type} $o_1 o_2 \ldots o_k$, $o_i \in O$, if for all levels $i$ the operators at level $i$ are equal to $o_i$. For example, $a a^+$ has type ``$\circ +$". 

%\todo{Need to define (non)reducibility. Asked a few people whether there is an abstract way to do it, we will see. For now, just the laundry list.}

To state our results, it is convenient to identify depth-3 homogeneous regular expressions that are equivalent to some depth-2 regular expressions.  
%Specifically, there are six depth-3 regular types that can be reduced to depth-2 types. 
In particular, in all types starting from $+$, the operator $+$ can be removed, since $p^+$ occurs in the text $t$ if and only if $p$ occurs in $t$. Similarly, in all types starting from $| +$, the operator $+$ can be removed, for the same reasons. \\
%Thus, we will refer to the latter 6 types as {\em reducible}, and to the remaining 6 types as {\em non-reducible}.\\
\paragraph{Notation} For symbol (or sequence) $z$ and integer $i$, $z^i$ denotes symbol (or sequence) $z$ repeated $i$ times. For an integer $n$, $[n]$ denotes $\{1,2,3,\ldots n\}$. Given an integer $d$, $0_d$ ($1_d$, resp.) denotes the vector with all entries equal to $0$ ($1$, resp.) in $d$ dimensions. For sequences $s_1,s_2,\ldots,s_k$, we write $\bigcirc_{i=1}^k s_i$ to denote the concatenation $s_1\,s_2\,\ldots\,s_k$ of the sequences.\\
\paragraph{Simplifying assumptions} To simplify our proofs, we will make several simplifying assumptions about the Orthogonal Vectors instance (different assumptions are used in different proofs). In what follows, we describe what kind of assumptions we make, and how to satisfy them:
\begin{enumerate}
\itemsep0em
\item The number of vectors $M$ in set $A$ is odd or even (depending on the proof):  this can be achieved w.l.o.g.  since we can add a vector to set $A$ consisting only of $1$s.  %The number of vectors $N$ in sets $A$ and $B$ is odd or even (depending on the proof):  this can be achieved w.l.o.g.  since we can add a vector to each set consisting only of $1$s. 
\item The dimensionality $d$ of vectors from sets $A$ and $B$ is odd or even (depending on the proof): this can be achieved w.l.o.g. since we can add new coordinate to every vector and set this coordinate to $0$. 
\item Let $A=\{a^1, \ldots, a^M\}$ and $B=\{b^1, \ldots, b^N\}$. \iffalse The two orthogonal vectors $a^i$ and $b^j$ satisfy $a^i\cdot b^j=0$ and $\mod{i}{j}{2}$ (assuming an orthogonal pair of vectors exists):\fi If there are $i,j$ such that $a^i\cdot b^j = 0$ then there are $i',j'$ such that $a^{i'}\cdot b^{j'} = 0$ and $\mod{i'}{j'}{2}$: this can be assumed w.l.o.g. since we can define  a set $A'=\{a^M,a^1,a^2,a^3,\ldots,a^{M-1}\}$ and perform two reductions, one on the pair of sets $A$ and $B$ and another one on the pair of sets $A'$ and $B$. If  a pair of orthogonal vectors exists, one of these reductions will detect it. 
\item The dimensionality $d$ is greater than $100$: we can assume this since otherwise Orthogonal Vectors problem can be solved in $O(2^d\cdot N)=O(N)$ time. 
\item  $b^t_1=b^t_d=0$ for all $t \in [N]$ and $d$ is odd: first, we make $d$ odd as described above. Then we add two entries for every vector from $A$ or $B$, one at the beginning and one at the end, and set both entries to  $0$. 
\item The first vector $a^1$ from the set $A$ is not orthogonal to all vectors from $B$:  first, we can detect whether this is the case in $O(dN)$ time. If $a^1$ is orthogonal to a vector from $B$, we have found a pair of orthogonal vectors. Otherwise we proceed with the reduction.
\end{enumerate}

%\todo{Explain why we can make various assumptions (indices = 0 modulo 2, n=0 modulo 2, etc)}
	
	\section{Reductions for the Pattern Matching problem}

We start this section by showing hardness for regular expressions of type ``$|{\circ}|$'' and ``$|{\circ}+$''. These hardness proofs are quite simple, and will help us introduce notation used in more complex reductions presented later. 
\iftoggle{cr}{After that}{Next} we present the hardness proof for regular expressions of type ``${\circ} *$''. 
\iftoggle{cr}{}{After that we present hardness proofs for the other cases.}
\iffalse
 for regular expressions of type ``${\circ}+{\circ}$''. Then we show how to modify this argument to give a hardness proof for regular expressions of type ``${\circ}|{\circ}$''. 
After that we present hardness proof for regular expressions of type ``${\circ}+|$''. Then we show how to modify this argument to give a hardness proof for regular expressions of type ``${\circ}|+$''.
\fi

\subsection{Hardness for type ``$|{{\circ}}|$''}
\label{oco}
	\begin{theorem}
		Given sets $A=\{a^1,\ldots,a^M\} \subseteq\{0,1\}^d$ and  $B=\{b^1,\ldots,b^N\} \subseteq\{0,1\}^d$, we can construct a regular expression $p$  and a sequence of symbols $t$, in $O(Nd)$ time, such that  a substring of $t$ can be derived from $p$ if and only if there are $a \in A$ and $b \in B$ such that $a \cdot b=0$.
		Furthermore, $p$ has type  ``$|{\circ}|$'', $|p|\leq O(Md)$ and $|t|\leq O(Nd)$.
	\end{theorem}
	\begin{proof}
		First, we will construct our pattern $p$.		
		For an integer $v \in \{0,1\}$, we construct the following pattern coordinate gadget
		$$
			CG(v):=
				\begin{cases}
					[0|1]	& \text{if }v=0;\\
					[0]	& \text{if }v=1.\\
				\end{cases}
		$$
		For a vector $a \in \{0,1\}^d$, we define a pattern vector gadget
		$$
			VG(a):=CG(a_1)\,CG(a_2)\,CG(a_3)\, \ldots \,CG(a_d).
		$$
		Our pattern $p$ is then defined as ``$|$'' of all pattern vector gadgets:
		\iffalse
		\begin{align*}
			p:=& VG\left(a^1\right)\,|\,VG\left(a^2\right)\,|\,VG\left(a^3\right)\,|\,VG\left(a^4\right)\,|\, \\
				& \ldots \,|\,VG\left(a^{M-1}\right)\,|\,VG\left(a^M\right).
		\end{align*}
		\fi
		$$ p:=VG\left(a^1\right)\,|\,VG\left(a^2\right)\,|\,VG\left(a^3\right)\,|\,VG\left(a^4\right)\,|\,
				\ldots \,|\,VG\left(a^{M-1}\right)\,|\,VG\left(a^M\right).$$

		Now we construct the text $t$. First, for a vector $b \in \{0,1\}^d$, we define text vector gadget as concatenation of all entries of $b$:
		$
			VG'(b):=b_1\,b_2\,b_3\,b_4\,\ldots\,b_d.
		$
		Note that we can derive $VG'(b)$ from $VG(a)$ if and only if $a \cdot b=0$.
		Our text $t$ is defined as a concatenation of all text vector gadgets with a symbol $2$ in between any two neighbouring vector gadgets:
		\iffalse
		\begin{align*}
			t:=& VG'\left(b^1\right)\,2\,VG'\left(b^2\right)\,2\,VG'\left(b^3\right)\,2\,VG'\left(b^4\right)\,2\, \\
			   & \ldots \,2\,VG'\left(b^{N-1}\right)\,2\,VG'\left(b^N\right).
		\end{align*}
		\fi
			$$t:=VG'\left(b^1\right)\,2\,VG'\left(b^2\right)\,2\,VG'\left(b^3\right)\,2\,VG'\left(b^4\right)\,2\,
			   \ldots \,2\,VG'\left(b^{N-1}\right)\,2\,VG'\left(b^N\right).$$

		We need to show that we can derive a substring of $t$ from $p$ if and only if there are two orthogonal vectors in $A$ and $B$. This follows from lemmas \ref{ort4} and \ref{nort4} below.
	\end{proof}

	\begin{lemma} \label{ort4}
		If there are two vectors $a \in A$ and $b \in B$ that are orthogonal, then a substring of $t$ can be derived from $p$.
	\end{lemma}
	\begin{proof}
		Suppose that $a^i \cdot b^j=0$ for some $i \in [M]$, $j \in [N]$. We choose a pattern vector gadget $VG(a^i)$ from the pattern $p$ and transform it into a text vector gadget $VG'(b^j)$. This is possible because of the orthogonality and the construction of the vector gadgets.
	\end{proof}

	\begin{lemma} \label{nort4}
		If a substring of $t$ can be derived from $p$, then there are two orthogonal vectors.
	\end{lemma}
	\begin{proof}
		By the construction of pattern $p$, we have to choose one pattern vector gadget, say, $VG(a^i)$, that is transformed into a binary substring of $t$ of length $d$. The text $t$ has the property that it is a concatenation of binary strings of length $d$ separated by symbols $2$. This means that $VG(a^i)$ will be transformed into binary string $VG'(b^j)$ for some $j$. This implies that $a^i \cdot b^j=0$ by the construction of the vector gadgets.
	\end{proof}

\subsection{Hardness for type ``$|{\circ}+$''}
\label{ocp}
	\begin{theorem}
		Given sets $A=\{a^1,\ldots,a^M\} \subseteq\{0,1\}^d$ and  $B=\{b^1,\ldots,b^N\} \subseteq\{0,1\}^d$, we can construct a regular expression $p$  and a sequence of symbols $t$, in $O(Nd)$ time, such that  a substring of $t$ can be derived from $p$ iff there are $a \in A$ and $b \in B$ such that $a \cdot b=0$.
		Furthermore, $p$ has type ``$|{\circ}+$'', $|p|\leq O(Md)$ and $|t|\leq O(Nd)$.
	\end{theorem}
	\begin{proof}
		First, we will construct our pattern.		
		For an integer $v \in \{0,1\}$, we construct the following pattern coordinate gadget
		$$
			CG(v):=
				\begin{cases}
					x^+		& \text{if }v=0;\\
					x^+x^+		& \text{if }v=1.
				\end{cases}
		$$
		For a vector $a \in \{0,1\}^d$, we define a pattern vector gadget as concatenation of coordinate gadgets for all coordinates with the symbol $y$ in between every two neighbouring coordinate gadgets:
		$$
			VG(a):=CG(a_1)\,y\,CG(a_2)\,y\,CG(a_3)\,y\, \ldots \,y\,CG(a_d).
		$$
		Our pattern $p$ is then defined as an OR (``$|$'') of all the pattern vector gadgets:
		\iffalse
		\begin{align*}
			p:=& VG\left(a^1\right)\,|\,VG\left(a^2\right)\,|\,VG\left(a^3\right)\,|\,VG\left(a^4\right)\,|\, \\
				& \ldots \,|\,VG\left(a^{M-1}\right)\,|\,VG\left(a^M\right).
		\end{align*}
		\fi
		$$ p:=VG\left(a^1\right)\,|\,VG\left(a^2\right)\,|\,VG\left(a^3\right)\,|\,VG\left(a^4\right)\,|\,
				 \ldots \,|\,VG\left(a^{M-1}\right)\,|\,VG\left(a^M\right).$$

		Now we proceed with the construction of our text $t$.
		For an integer $v \in \{0,1\}$, we define the following text coordinate gadget
		$$
			CG'(v):=
				\begin{cases}
					xx		& \text{if }v=0;\\
					x		& \text{if }v=1.
				\end{cases}
		$$
		For vector $b \in \{0,1\}^d$, we define the text vector gadget as
		\iffalse
		\begin{align*}
			VG'(b):=& CG'(b_1,1)\,y\,CG'(b_2,2)\,y\,CG'(b_3,3)\,y\, \\
					& \ldots \,y\,CG'(b_d,d).
		\end{align*}
		\fi
		$$ VG'(b):=CG'(b_1,1)\,y\,CG'(b_2,2)\,y\,CG'(b_3,3)\,y\,
					\ldots \,y\,CG'(b_d,d).$$
		Note that we can derive $VG'(b)$ from $VG(a)$ iff $a \cdot b=0$.
		Our text $t$ is defined as a concatenation of all text vector gadgets with the symbol $z$ in between any two neighbouring vector gadgets:
		\iffalse
		\begin{align*}
			t:=& VG'(b^1)\,z\,VG'(b^2)\,z\,VG'(b^3)\,z\,VG'(b^4)\,z\, \\
				& \ldots \,z\,VG'(b^{N-1})\,z\,VG'(b^N).
		\end{align*}
		\fi
		$$t:=VG'(b^1)\,z\,VG'(b^2)\,z\,VG'(b^3)\,z\,VG'(b^4)\,z\,
				\ldots \,z\,VG'(b^{N-1})\,z\,VG'(b^N).$$

		We need to show that we can derive a substring of $t$ from $p$ iff there are two orthogonal vectors. This follows from lemmas \ref{ort5} and \ref{nort5} below. 
	\end{proof}

	\begin{lemma} \label{ort5}
		If there are two vectors $a \in A$ and $b \in B$ that are orthogonal, then a substring of $t$ can be derived from $p$.
	\end{lemma}
	\begin{proof}
		Suppose that $a^i \cdot b^j=0$ for some $i \in [M]$, $j \in [N]$. We choose a pattern vector gadget $VG(a^i)$ from the pattern $p$ and transform it into a text vector gadget $VG'(b^j)$. This is possible because of the orthogonality and the construction of the vector gadgets.
	\end{proof}

	\begin{lemma} \label{nort5}
		If a substring of $t$ can be derived from $p$, then there are two orthogonal vectors.
	\end{lemma}
	\begin{proof}
		We call a sequence of symbols \emph{nice} iff it can be derived from the regular expression $x^+yx^+yx^+y\ldots yx^+$, where ``$x^+$'' appears $d$ times.

		By the construction of pattern $p$, we have to choose one pattern vector gadget, say, $VG(a^i)$, that is transformed into a nice sequence. The text $t$ has the property that it is a concatenation of nice sequences separated by symbols $z$. This means that $VG(a^i)$ will be transformed into a sequence $VG'(b^j)$ for some $j$. This implies that $a^i \cdot b^j=0$ by the construction of vector gadgets.
	\end{proof}

\subsection{Hardness for type ``${\circ} *$''}
\label{cs}
	\begin{theorem} \label{cs_pm}
		Given sets $A=\{a^1,\ldots,a^M\} \subseteq\{0,1\}^d$ and $B=\{b^1,\ldots,b^N\} \subseteq\{0,1\}^d$  with $M\leq N$, we can construct the regular expression $p$ and  the text $t$ in time $O(Nd)$, such that
		 a substring of $t$ can be derived from $p$ iff there are $a \in A$ and $b \in B$ such that $a \cdot b=0$.
		Furthermore,  $p$ is of type ``${\circ} *$", $|p|\leq O(Md)$ and $|t|\leq O(Nd)$.
	\end{theorem}
	\begin{proof}
		W.l.o.g., we can assume that $\mod{M}{1}{2}$ and $\mod{d}{1}{2}$, $d\geq 100$.
		Also, if there are $i \in [M]$, $j\in[N]$ such that $a^i \cdot b^j=0$, then there are $i' \in [M]$, $j'\in[N]$ such that $a^{i'} \cdot b^{j'}=0$ and $\mod{i'}{j'}{2}$. Furthermore, we assume $b^j_1=b^j_d=0$ for all $j \in [N]$ and that $a^1$ is not orthogonal to any vector $b^j$.

		First, we will construct our pattern.		
		For an integer $v \in \{0,1\}$ and an integer $i \in [d]$, we construct the following pattern coordinate gadget
		$$
			CG(v,i):=
				\begin{cases}
					yy^*		& \text{if }v=0\text{ and }\mod{i}{1}{2};\\
					yyyy^*		& \text{if }v=1\text{ and }\mod{i}{1}{2};\\
					xx^*		& \text{if }v=0\text{ and }\mod{i}{0}{2};\\
					xxxx^*		& \text{if }v=1\text{ and }\mod{i}{0}{2}.
				\end{cases}
		$$
		For a vector $a \in \{0,1\}^d$, we define a pattern vector gadget
		$$
			VG(a):=CG(a_1,1)\,CG(a_2,2)\,CG(a_3,3)\, \ldots \,CG(a_d,d).
		$$
		We also need another pattern vector gadget $VG_0:=(y^*\,x^*)^{d+10}\,y^*$.

		Our pattern is then defined as follows:
		$$
			p:=y^6\, \bigcirc_{j \in [M-1]}\left(x^{10}\,VG(a^j)\,x^{10}\,VG_0\right)\,x^{10}\,VG(a^M)\,x^{10}\,y^6.
		$$

		Now we proceed with the construction of our text.
		For integers $v \in \{0,1\}, i \in [d]$, we define the following text coordinate gadget
		$$
			CG'(v,i):=
				\begin{cases}
					yyy		& \text{if }v=0\text{ and }\mod{i}{1}{2};\\
					y		& \text{if }v=1\text{ and }\mod{i}{1}{2};\\
					xxx		& \text{if }v=0\text{ and }\mod{i}{0}{2};\\
					x		& \text{if }v=1\text{ and }\mod{i}{0}{2}.
				\end{cases}
		$$
		For a vector $b \in \{0,1\}^d$ and an integer $\mod{j}{1}{2}$, we define the text vector gadget as
		$$
			VG'(b,j):=CG'(b_1,1)\,CG'(b_2,2)\,CG'(b_3,3)\,\ldots \,CG'(b_d,d).
		$$
		We also define $VG'(b,j)$, when $\mod{j}{0}{2}$. In this case, $VG'(b,j)$ is equal to $VG'(b,1)$ except that we replace every occurrence of the substring $y^3$ with the substring $y^6$.
		
		One can verify that for any {\em vectors} $a,b \in \{0,1\}^d$ and any integer $i$, $VG'(b,i)$ can be derived from $VG(a)$ iff $a \cdot b=0$.

		We will also need an additional text vector gadget
		$$
			VG_0':=y^3\,(x^3\,y^3)^{(d-1)/2}.
		$$

		Our text is then defined as follows:
		$$
			t:=\bigcirc_{j=-2N}^{3N}\left(x^{10}\,VG_0'\,x^{10}\,VG'(b^j,j)\right),
		$$
		where we assume $b^j:=011111\ldots 111110$ for $j \not \in [N]$.

		We have to show that we can derive a substring of $t$ from $p$ iff there are two orthogonal vectors. This follows from lemmas \ref{ort} and \ref{nort} below.
	\end{proof}

	\begin{lemma} \label{ort}
		If there are two vectors $a \in A$ and $b \in B$ that are orthogonal, then a substring of $t$ can be derived from $p$.
	\end{lemma}
	\begin{proof}
		W.l.o.g. we have that, $a^k \cdot b^k=0$ for some $k \in [M]$. The proof for the case when $a^k \cdot b^r=0$, $k \in [M]$, $r \in [N]$, $\mod{k}{r}{2}$ is analogous.
		
		The pattern $p$ starts with $y^6$. We transform it into $CG'(b^0_d,d)$ appearing in $VG'(b^0,0)$. We can do this since $b^0_d=0$.

		For $j=1,2,\ldots, k-2$, we transform $x^{10}\,VG(a^j)\,x^{10}\,VG_0$ into $x^{10}\,VG_0'\,x^{10}\,VG'(b^j,j)$ by transforming $VG(a^j)$ into $VG_0'$ and $VG_0$ into $VG'(b^j,j)$.

		Next we transform 
		$$
			x^{10}\,VG(a^{k-1})\,x^{10}\,VG_0\,x^{10}\,VG(a^{k})\,x^{10}\,VG_0
		$$ 
		into 
		\begin{align*}
			& x^{10}\,VG_0'\,x^{10}\,VG'(b^{k-1},k-1)\,x^{10}\,VG_0'\,x^{10}\, \\
			& VG'(b^{k},k)\,x^{10}\,VG_0'\, x^{10}\,VG'(b^{k+1},k+1)
		\end{align*}
		Notice that we use the fact that $k\geq 2$ (we assumed that $a^1$ is not orthogonal to any vector from $B$). Note that $VG'(b^{k+1},k+1)$ appears in the text $t$ even if $k=N$. This is because in the definition of the text $t$, integer $j$ ranges from $-2N$ up to $3N$.

		We perform the transformation by performing the following steps:
		\begin{enumerate}
			\item transform $VG(a^{k-1})$ into $VG'_0$;
			\item transform $VG_0$ into $VG'(b^{k-1},k-1)\,x^{10}\,VG'_0$;
			\item transform $VG(a^k)$ into $VG'(b^k,k)$ (we can do this since $a^k\cdot b^k=0$);
			\item transform $VG_0$ into $VG_0'\,x^{10}\,VG'(b^{k+1},k+1)$.
		\end{enumerate}

		Now, for $j=k+1,\ldots,M-1$ transform $x^{10}\,VG(a^j)\,x^{10}\,VG_0$ into $x^{10}\,VG_0'\,x^{10}\,VG'(b^{j+1},j+1)$ similarly as before.
		Next, transform $x^{10}\,VG(a^M)\,x^{10}$ into $x^{10}\,VG_0'\,x^{10}$. Finally, transform $y^6$ into $CG'(b^{M+1}_1)$ appearing in $VG'(b^{M+1},M+1)$. We can do this since $b^{M+1}_1=0$.
	\end{proof}

	\begin{lemma} \label{nort}
		If a substring of $t$ can be derived from $p$, then there are two orthogonal vectors.
	\end{lemma}
	\begin{proof}
		By the construction, every substring $x^{10}$ from $p$ must be mapped to a unique substring $x^{10}$ in $t$ (there are no substrings of $t$ that have more than $10$ symbols $x$). Because of this, every $VG(a^i)$ must be mapped to $VG'_0$ or $VG'(b^j,j)$ for some $j$. If the latter case occurs, the corresponding vectors are orthogonal and we are done. It remains to consider the case that \emph{all} vector gadgets $VG(a^i)$ get mapped to $VG'_0$. Consider any vector gadget $VG_0$ in $p$. To the left of it we have the sequence $x^{10}$ and to the right of it we have the sequence $x^{10}$. Each one of these two sequences $x^{10}$ in $p$ gets mapped to a unique sequence $x^{10}$ in $t$. We call the vector gadget $VG_0$ \emph{nice} if the two unique sequences $x^{10}$ are neighbouring in $t$, that is, there is no other sequence $x^{10}$ in $t$ between the two unique sequences. We consider two cases below.

\paragraph{Case $1$} There is a vector gadget $VG_0$ in $p$ that is \emph{not} nice. Take any vector gadget $VG_0$ that is not nice and denote it by $v$. The gadget $v$ is immediately to the right of the expression $VG(a^{i'})\,x^{10}$ in $p$ for some $i'\in [M]$. $VG(a^{i'})$ is mapped to $VG'_0$ (otherwise, we have found an orthogonal pair of vectors, as per the discussion above) and this $VG'_0$ is to the left of substring $x^{10}\,VG'(b^{j''},j'')$ in $t$ for some $j''$. Because $v$ is \emph{not} nice, a prefix of it must map to $VG'(b^{j''},j'')\,x^{10}\,VG'_0$. We claim that \emph{entire} $v$ gets mapped to $VG'(b^{j''},j'')\,x^{10}\,VG'_0$. If this is not the case, then a prefix of $v$ must be mapped to $VG'(b^{j''},j'')\,x^{10}\,VG'_0\,x^{10}\,VG'(b^{j''+1},j''+1)$ (since sequence $x^{10}$ in $p$ to the right of $v$ must be mapped to $x^{10}$). A prefix of $v$ can't be mapped to $VG'(b^{j''},j'')\,x^{10}\,VG'_0\,x^{10}\,VG'(b^{j''+1},j''+1)$ since $v=(y^*\,x^*)^{d+10}y^*$ can produce sequence with at most $d+11$ substrings of maximal length consisting entirely of symbols $y$ but sequence $VG'(b^{j''},j'')\,x^{10}\,VG'_0\,x^{10}\,VG'(b^{j''+1},j''+1)$ has $3{d+1 \over 2}>d+11$ (if $d\geq 100$) subsequences of maximal length consisting entirely of $y$. Therefore, we are left with the case that entire $v$ is mapped to $VG'(b^{j''},j'')\,x^{10}\,VG'_0$. The gadget $v$ is to the left of vector gadget $VG(a^{i'+1})$ and $VG'(b^{j''},j'')\,x^{10}\,VG'_0$ is to the left of vector gadget $VG'(b^{j''+1},j''+1)$. We conclude that $VG(a^{i'+1})$ must be mapped to $VG'(b^{j''+1},j''+1)$. This implies that $a^{i'+1}\cdot b^{j''+1}=0$ and we are done.

\paragraph{Case $2$} All vector gadgets $VG_0$ in $p$ are nice. We start $p$ with $y^6$ followed immediately by $x^{10}$. This means that $y^6$ is mapped to $CG'(b^{j'}_d,j')$ for some \emph{even} $j'$ (by the construction of coordinate gadgets $CG'$). Consider vector gadgets in $p$ from left to the right. We must have that $VG(a^1)$ is mapped to $VG'_0$, that $VG_0$ is mapped to $VG'(b^1,1)$ (since every $VG_0$ in $p$ is nice), that $VG(a^2)$ is mapped to $VG'_0$, that $VG_0$ is mapped to $VG'(b^2,2)$ (since every $VG_0$ in $p$ is nice) and so forth. Since $M$ is odd, we have that $VG(a^M)$ is mapped to $VG'_0$ and that this vector gadget $VG'_0$ is followed by $x^{10}\,VG'(b^{j'+M},j'+M)$. $VG(a^M)$ is followed by $x^{10}y^6$ and this means that $y^6$ is mapped to the beginning of $VG'(b^{j'+M},j'+M)$. This is impossible since $VG'(b^{j'+M},j'+M)$ does not contain a substring of length $6$ or more consisting of symbols $y$ (observe that $j'+M$ is odd and see the construction of vector gadget $VG'$). We get that Case $2$ can't happen.
\end{proof}

\iffalse
	We start this section by showing hardness for regular expressions of type ``$|{\circ}|$'' and ``$|{\circ}+$''. These hardness proofs are quite simple and will help us introduce notation used in more complex reductions presented later. In Subsection \ref{cpo} we provide a hardness proof for regular expressions of type ``${\circ}+|$''. In the following subsection we show how we can adapt this argument so that it gives a hardness proof for regular expressions of type ``${\circ}|+$''. In Subsection \ref{cpc} we show hardness for regular expressions of type ``${\circ}+{\circ}$''. In Subsection \ref{coc}, we show how to modify this argument to give a hardness proof for regular expressions of type ``${\circ}|{\circ}$''. Finally, in the last subsection we give a hardness proof for type ``${\circ} *$''.

\fi

\subsection{Hardness for type ``${\circ}{+}{\circ}$''}
 \label{cpc}
	\begin{theorem} \label{cpct}
		Given sets $A=\{a^1,\ldots,a^M\} \subseteq\{0,1\}^d$ and  $B=\{b^1,\ldots,b^N\} \subseteq\{0,1\}^d$  with $M\leq N$, we can construct a regular expression $p$  and a sequence of symbols $t$,  in $O(Nd)$ time, such that  a substring of $t$ can be derived from $p$ if and only if there are $a \in A$ and $b \in B$ such that $a \cdot b=0$.
		Furthermore, $p$ is a concatenation of ``$+$'' of sequences, $|p|\leq O(Md)$ and $|t|\leq O(Nd)$.
	\end{theorem}
	\begin{proof}
		W.l.o.g., $d$ is even.		
		First, we will construct our pattern $p$.		
		For an integer $v \in \{0,1\}$ and an integer $i \in [d]$ (remember that $[d]=\{1,2,\ldots,d\}$), we construct the following pattern coordinate gadget
		$$
			CG(v,i):=
				\begin{cases}
					[x]^+		& \text{if }v=0\text{ and }\mod{i}{1}{2};\\
					[xx]^+		& \text{if }v=1\text{ and }\mod{i}{1}{2};\\
					[y]^+		& \text{if }v=0\text{ and }\mod{i}{0}{2};\\
					[yy]^+		& \text{if }v=1\text{ and }\mod{i}{0}{2}.
				\end{cases}
		$$
		For a vector $a \in \{0,1\}^d$, we define a pattern vector gadget as $[x^4]^+\,[y^4]^+$ followed by the concatenation of all coordinate gadgets:
		$$
			VG(a):=[x^4]^+\,[y^4]^+\,CG(a_1,1)\,CG(a_2,2)\,CG(a_3,3)\, \ldots \,CG(a_d,d).
		$$
		We also need two other pattern vector gadgets
		$$
			VG_0:=[x^4\,y^4]^+\,([x]^+\,[y]^+)^{d/2};
			\hspace{1.5 cm}
			VG_1:=[x^4]^+\,[y^4]^+\,([x]^+\,[y^8]^+)^{d/2}.
		$$
% PIOTR: confusing
%		By our notational assumption, $([x]^+\,[y^8]^+)^{d/2}$, for example, denotes sequence $[x]^+\,[y^8]^+\,[x]^+\,[y^8]^+\,\ldots\,[x]^+\,[y^8]^+$, where $[x]^+$ appears $d/2$ times.
		Our pattern is then defined as follows:
		$$
			p:=VG_1\,\left(\bigcirc_{j \in [M]}\left(VG(1_d)\,VG_0\,VG(a^j)\,VG_0\,\right)\right)\,VG(1_d)\,VG(1_d)\,VG_1.
		$$

		Now we proceed with the construction of our text.
		For integers $v \in \{0,1\}, i \in [d]$, we define the following text coordinate gadget
		$$
			CG'(v,i):=
				\begin{cases}
					xx		& \text{if }v=0\text{ and }\mod{i}{1}{2};\\
					x		& \text{if }v=1\text{ and }\mod{i}{1}{2};\\
					yy		& \text{if }v=0\text{ and }\mod{i}{0}{2};\\
					y		& \text{if }v=1\text{ and }\mod{i}{0}{2}.
				\end{cases}
		$$
		For vector $b \in \{0,1\}^d$, we define the text vector gadget as $x^4\,y^4$ followed by the concatenation of all coordinate gadgets:
		$$
			VG'(b):=x^4\,y^4\,CG'(b_1,1)\,CG'(b_2,2)\,CG'(b_3,3)\,\ldots \,CG'(b_d,d).
		$$
		
		In what follows, we will use the following important property of vector gadgets $VG$ and $VG'$. First, observe that for integers $u,v\in\{0,1\},i \in [d]$, we can derive $CG'(v,i)$ from $CG(v,i)$ if and only if $uv=0$. 
		It means that for any {\em vectors} $a,b \in \{0,1\}^d$, $VG'(b)$ can be derived from $VG(a)$ if and only if $a \cdot b=0$.

		We will also need additional text vector gadgets
		$$
			VG_0':=x^4\,y^4\,(x^4\,y^4)^{d/2};
			\hspace{1.5 cm}
			VG_1':=x^4\,y^4\,(x\,(y)^8)^{d/2}.
		$$

		Our text is then defined as follows:
		$$
			t:=\bigcirc_{j=-5N}^{6N}\left(VG'(0_d)\,VG_0'\,VG'(b^j)\,VG'(0_d)\,VG_0'\,VG_1'\right),
		$$
		where we assume $b^j:=1_d$ for $j \not \in [N]$. That is, $b^j$ is vector with $d$ entries all equal to $1$ for $j \not \in [N]$.

		We have to show that we can derive a substring of $t$ from $p$ if and only if there are two orthogonal vectors in $A$ and $B$. This follows from lemmas \ref{ort2} and \ref{nort2} below. 
		%The rest of the properties follows from the construction.
	\end{proof}

	\begin{lemma} \label{ort2}
		If there are two vectors $a \in A$ and $b \in B$ that are orthogonal, then a substring of $t$ can be derived from $p$.
	\end{lemma}
	\begin{proof}
		We assume that $a^k \cdot b^r=0$ for some $k \in [M]$, $r \in [N]$.
		%We assume that $a^k \cdot b^k=0$ for some $k \in [M]$. In the general case when $a^k \cdot b^r=0$ for $k \in [M]$, $r \in [N]$ and $k \neq r$, the proof is analogous. 
		
		Observe that the pattern $p$ starts with a prefix $VG_1$, which is then followed by the sequence
		$$
			VG(1_d)\,VG_0\,VG(a^i)\,VG_0
		$$
		repeated $M$ times for different $i \in [M]$ and ends with the suffix 
		$
			VG(1_d)\,VG(1_d)\,VG_1.
		$ 
		We refer to
		$$
			VG(1_d)\,VG_0\,VG(a^i)\,VG_0
		$$ 
		corresponding to a specific $i \in [M]$ as the $i$-th {\em group} of $p$. Similarly, we observe that the text $t$ is concatenation of sequences 
		$$
			VG'(0_d)\,VG_0'\,VG'(b^j)\,VG'(0_d)\,VG_0'\,VG_1'
		$$ 
		corresponding to different $j=-5N,\ldots,6N$. We refer to that sequence corresponding to particular $j=-5N,\ldots,6N$ as the $j$-th  {\em group} of $t$.

		We transform $p$ into the following substring $t'$ of $t$:
		$$
			t':=VG_1'\,\bigcirc_{j=1+r-k}^{M+r-k}\left(VG'(0_d)\,VG_0'\,VG'(b^j)\,VG'(0_d)\,VG_0'\,VG_1'\right).
		$$
		Note that $t'$ starts with the prefix $VG_1'$ which is a suffix of the $(r-k)$-th group of $t$. Then it consists of groups $1+r-k$ to $M+r-k$ of $t$. 
		We transform $p$ into $t'$ from left to right.  First, we transform $VG_1$ into $VG_1'$ in the unique way.
		Then, for $i=1,\ldots,k-1$, we transform the $i$-th group of $p$ into the $(i+r-k)$-th group of $t'$. We do that as follows. For each $i \in \{1,\ldots,k-1\}$, we perform the following transformations: transform $VG(1_d)$ into $VG'(0_d)$; transform $VG_0$ into $VG_0'\,VG'(b^{i+r-k})$; transform $VG(a^i)$ into $VG'(0_d)$; transform $VG_0$ into $VG_0'\,VG_1'$.

		Now we transform  the $k$-th and the $(k+1)$-th group of $p$ into the prefix 
		$$
			VG'(0_d)\,VG_0'\,VG'(b^r)\,VG'(0_d)\,VG_0'\,VG_1'\,VG'(0_d)\,VG_0'\,VG'(b^{r+1})
		$$ 
		of the remainder of $t'$. This is done by transforming $VG(1_d)$ into $VG'(0_d)$, $VG_0$ into $VG_0'$, $VG(a^k)$ into $VG'(b^r)$ (which can be done because $a^k\cdot b^r=0$), $VG_0$ into $VG'(0_d)$, $VG(1_d)$ into $VG_0'$, $VG_0$ into $VG_1'$, $VG(a^{k+1})$ into $VG'(0_d)$ and $VG_0$ into $VG_0'\,VG'(b^{r+1})$. The remainder of the pattern $p$ consists of groups $k+2,\ldots,M$ and the suffix $VG(1_d)\,VG(1_d)\,VG_1$, while the remainder of $t'$ is
		$$
			t'':=\left(\bigcirc_{j=r+2}^{M+r-k}\left(VG'(0_d)\,VG_0'\,VG_1'\,VG'(0_d)\,VG_0'\,VG'(b^j)\right)\right)\,VG'(0_d)\,VG_0'\,VG_1'.
		$$
		We transform group $i=k+2,\ldots,M$ of $p$ into a substring $VG'(0_d)\,VG_0'\,VG_1'\,VG'(0_d)\,VG_0'\,VG'(b^{i+r-k})$, as follows. We transform $VG(1_d)$ into $VG'(0_d)$, $VG_0$ into $VG_0'\,VG_1'$, $VG(a^i)$ into $VG'(0_d)$ and $VG_0$ into $VG_0'\,VG'(b^{i+r-k})$.
		It remains to transform the suffix $VG(1_d)\,VG(1_d)\,VG_1$ of $p$ into the suffix $VG'(0_d)\,VG_0'\,VG_1'$ of $t''$, which can be done in a non-ambiguous way.
	\end{proof}

	\begin{lemma} \label{nort2}
		If a substring of $t$ can be derived from $p$, then there are two vectors $a \in A$ and $b \in B$ that are orthogonal.
	\end{lemma}
	\begin{proof}
		Notice that $p$ starts and ends with $VG_1$ and that this vector gadget does not appear anywhere else in the pattern.
		$VG_1$ contains $d/2$ patterns $[y^8]^+$. This means that it must map to a substring of $t$ containing $d/2$ substrings consisting of only symbols $y$ of length divisible by $8$. Text $t$ contains vector gadgets $VG_1'$ that has substrings of symbols $y$ of length divisible by $8$. No other vector gadget in $t$ has this property. Therefore, both vector gadgets $VG_1$ of $p$ must map to $VG_1'$ in $t$. Consider $VG_1$ at the beginning of $p$. Suppose that it maps to $VG_1'$ that comes from the $r$-th group of $t$, for some $r$. The pattern $p$ contains $2M$ vector gadgets $VG_0$. Each one of them starts with $[x^4\,y^4]^+$. In the proof of Lemma \ref{ort2} we transform $[x^4\,y^4]^+$ into $x^4\,y^4$ so that $VG_0$ maps to a single corresponding vector gadget in $t$ or we transform $[x^4\,y^4]^+$ into $(x^4\,y^4)^{2+d/2}$ so that $VG_0$ maps to two vector gadgets in $t$. Here we consider the first copy $v$ of the vector gadget $VG_0$ of $p$ such that the corresponding $[x^4\,y^4]^+$ \emph{does not} map to $(x^4\,y^4)^{2+d/2}$. If there is no such a vector gadget, we can check that the $j$-th group of $p$ must be transformed into the $j+r$-th group of $t$ for all $t \in [M]$. That means that the suffix $VG(1_d)\,VG(1_d)\,VG_1$ of $p$ must be transformed into the prefix of the $M+r+1$-st group of $t$. This implies that $VG_1$ is transformed into $VG'(b^{M+r+1})$, which is impossible by the construction of the vector gadgets. Therefore, there must be vector gadget $v$ with the stated properties.
		
Suppose that $v$ comes from the $i$-th group in $p$. For all vector gadgets of type $VG_0$ to the left of $v$, the prefix $[x^4\,y^4]^+$ got mapped to $(x^4\,y^4)^{2+d/2}$. This means that group $k=1,\ldots, i-1$ of $p$ is transformed into the group $k+r$ of $t$. There are two cases to consider, depending on whether $v$ comes before or after the vector gadget $VG(a^i)$ in the group $i$ of $p$.

	\paragraph{Case 1: $v$ comes after $VG(a^i)$} We show that this case cannot occur. We start by observing that $VG(a^i)$ maps to $VG'(0_d)$.  The sequence $[x^4\,y^4]^+$ in $v$ cannot map to $3+d/2$ or more copies of $x^4\,y^4$ because it would imply that the first substring $x\,(y)^8$ of $VG_1'$ can be derived from $[x^4\,y^4]^+$, which is impossible. It follows that $[x^4\,y^4]^+$ maps to $1+d/2$ or fewer copies of $x^4\,y^4$. Note that $v$ is followed by $VG(1_d)$ and $VG(1_d)$ starts with $[x^4]^+\,[y^4]^+$.
This means that $[x^4\,y^4]^+$ of $v$ must map to only one copy of $x^4\,y^4$ (i.e., $VG_0$ is mapped to $VG'_0$), so that $[x^4]^+\,[y^4]^+$ at the beginning of $VG(1_d)$ can map to $x^4 \, y^4$ at the beginning of $VG_1'$. This in turn implies that $VG(1_d)$ must be transformed into a suffix of $VG_1'$, which is again impossible since $VG_1'$ contain a single symbol $x$ surrounded by symbols $y$ on both sides and the regular expression $VG(1_d)$ can't produce such a substring. Thus, Case 1 cannot occur.

	\paragraph{Case 2: $v$ comes before $VG(a^i)$} The argument is similar to the previous paragraph except we will conclude that $a^i \cdot b^{i+r}=0$. Similarly as before, we can conclude that $[x^4\,y^4]^+$ in $v$ cannot be mapped to $3+d/2$ or more copies of $x^4\,y^4$. This is because $VG'(b^{i+r})$ starts with $x^4\,y^4\,CG'(b^{i+r}_1,1)$ and, if $[x^4\,y^4]^+$ in $v$ was mapped to $3+d/2$ or more copies of $x^4\,y^4$, $CG'(b^{i+r}_1,1)$ would be a prefix of sequence that can be derived from $[x^4\,y^4]^+$. We can verify that it is impossible. Assume that $[x^4\,y^4]^+$ maps to $1+d/2$ or fewer copies of $x^4\,y^4$. Because $VG(a^i)$ starts with $[x^4]^+\,[y^4]^+$, we have that $VG(a^i)$ is transformed into $VG'(b^{i+r})$. This is possible only if $a^i \cdot b^{i+r}=0$ by the construction of the coordinate gadgets.
	\end{proof}

\subsection{Hardness for type ``${\circ}|{\circ}$''} \label{coc}
	\begin{theorem} \label{coc_pm}
		Given sets $A=\{a^1,\ldots,a^M\} \subseteq\{0,1\}^d$ and  $B=\{b^1,\ldots,b^N\} \subseteq\{0,1\}^d$ with $M\leq N$, we can construct a regular expression $p$  and a sequence of symbols $t$,  in $O(Nd)$ time, such that  a substring of $t$ can be derived from $p$ if and only if there are $a \in A$ and $b \in B$ such that $a \cdot b=0$. Furthermore, $p$ is of type ``${\circ}|{\circ}$", $|p|\leq O(Md)$ and $|t|\leq O(Nd)$.
	\end{theorem}
	\begin{proof}
		We will modify the construction for ``${\circ}{+}{\circ}$'' so that it gives a hardness proof for ``${\circ}|{\circ}$''. The text $t$ remains the same. We will modify pattern $p$ as follows. First, recall that $p$ is a repeated concatenation of regular expressions of the form $[x]^+,[x\,x]^+,[y]^+,[y\,y]^+,[x^4]^+,[y^4]^+,[x^4\,y^4]^+,[y^8]^+$ in some order. In the proof of Lemma \ref{ort2}, all of those sequences get repeated at most $8$ times, except for the sequence $x^4\,y^4$ which gets repeated once or $2+d/2$ times. Therefore, we replace $[s]^+$ with 
	$
		[s\,|\,s^2\,|\,s^3\,|\,s^4\,|\,s^5\,|\,s^6\,|\,s^7\,|\,s^8]
	$
	for all $s$, except when $s=x^4\,y^4$. In the latter case we replace $[x^4\,y^4]^+$ with
	$
		[x^4\,y^4\,|\,(x^4\,y^4)^{2+d/2}].
	$
		The proof of Lemma \ref{ort2} goes through as before. The proof of Lemma \ref{nort2} also goes through because, after these modifications, if some sequence $z$ can be derived from the modified pattern, it can be also derived from the initial pattern.
	\end{proof}

	\subsection{Hardness for type ``${\circ} {+}|$''} \label{cpo}
	\begin{theorem} \label{cpo_theorem}
		Given sets $A=\{a^1,\ldots,a^M\} \subseteq\{0,1\}^d$ and  $B=\{b^1,\ldots,b^N\} \subseteq\{0,1\}^d$  with $M\leq N$, we can construct a regular expression $p$  and a sequence of symbols $t$, in $O(Nd)$ time, such that  a substring of $t$ can be derived from $p$ iff there are $a \in A$ and $b \in B$ such that $a \cdot b=0$.
		Furthermore $p$ has type ``${\circ} {+}|$'', $|p|\leq O(Md)$ and $|t|\leq O(Nd)$.
	\end{theorem}
	\begin{proof}
		W.l.o.g., $\mod{M}{0}{2}$ and $\mod{d}{0}{2}$. W.l.o.g., if there are two orthogonal vectors, then there are $a^i \in A$, $b^j \in B$ with $a^i \cdot b^j=0$ and $i\equiv j(\text{mod }2)$.

		First, we will construct our pattern. We need the following coordinate gadget construction.
		For an integers $v \in \{0,1\}$ and $i \in [d]$,
		$$
			CG(v,i):=
				\begin{cases}
					[x\,|\,y\,|\,0\,|\,1]^+		& \text{if }v=0\text{ and }\mod{i}{1}{2};\\
					[x\,|\,y\,|\,0]^+					& \text{if }v=1\text{ and }\mod{i}{1}{2};\\
					[x\,|\,y\,|\,0'\,|\,1']^+		& \text{if }v=0\text{ and }\mod{i}{0}{2};\\
					[x\,|\,y\,|\,0']^+					& \text{if }v=1\text{ and }\mod{i}{0}{2}.
				\end{cases}
		$$
		For a vector $a \in \{0,1\}^d$, we define pattern vector gadget $VG(a)$ as concatenation of all coordinate gadgets for entries of the vector:
		$$
			VG(a):=CG(a_1,1)\,CG(a_2,2)\,CG(a_3,3)\, \ldots \,CG(a_d,d).
		$$
		We also need another vector gadget
		$$
			VG_0:=([0|1]^+\,[0'|1']^+)^{d/2},
		$$
		that is, $VG_0$ is equal to the vector gadget $VG(0_d)$ except it can't produce symbols $x$ and $y$.
		Our pattern $p$ is then defined as follows:
		$$
			p:=x^+\,\,\left(\bigcirc_{i \in [M]}\left(VG_0\,[x|y]^+\,VG(a^i)\,[x|y]^+\right)\right)\,VG_0\,\,x^+.
		$$
		
		Now we construct our text $t$. First, for a vector $b \in \{0,1\}^d$, we define text vector gadget
		$$
			VG'(b):=b_1\,b'_2\,b_3\,b'_4\,b_5\,\ldots\,b'_d,
		$$
		that is, it is concatenation of all entries and we put $'$ for every second entry.
		Note that we can derive $VG'(b)$ from $VG(a)$ iff $a \cdot b=0$. Also, we can derive $VG'(b)$ from $VG_0$ for any $b$.
		Our text $t$ is defined as
		\begin{align*}
			t:=\bigcirc_{j=-9N}^{10N}\left(x^{d+10}\,VG'(b^{2j})\,y^{d+10}\,VG'(b^{2j+1})\right),
		\end{align*}
		where, for $j \not \in [N]$, we set $b^j:=1_d$ (vector consisting of only $1$s). 
		
		We need to show that we can derive a substring of $t$ from $p$ iff there are two orthogonal vectors. This follows from lemmas \ref{ort3} and \ref{nort3} below.
	\end{proof}

	\begin{lemma} \label{ort3}
		If there are two vectors $a \in A$ and $b \in B$ that are orthogonal, then a substring of $t$ can be derived from $p$.
	\end{lemma}
	\begin{proof}
		We assume that $a^k \cdot b^k=0$ for some $k \in [M]$. In the general case when $a^k \cdot b^r=0$ for $k \in [M]$, $r \in [N]$ and $\mod{k}{r}{2}$, the proof is analogous. Furthermore, we assume that $\mod{k}{1}{2}$ (the case $\mod{k}{0}{2}$ is analogous).
		
		We transform $p$ into the following substring $t'$ of $t$:
		\begin{align*}
	t':=&x\,VG'(b^0)\,y^{d+10}\,VG'(b^1)\,x^{d+10}\,VG'(b^2)\,y^{d+10}\,VG'(b^3)\,x^{d+10}\,VG'(b^4)\,y^{d+10}\,VG'(b^5)\,x^{d+10} \\
			&\ldots \,x^{d+10}\,VG'(b^M)\,y^{d+10}\,VG'(b^{M+1})\,x.
		\end{align*}
		Note that $t'$ starts and ends with $x$,  because $\mod{M}{0}{2}$.

		To show how to transform $p$ into $t'$, it's helpful to write pattern $p$ as $p=p''\,p'''$, where
		\begin{align*}
			p'':=&x^+\,VG_0\,[x|y]^+\,VG(a^1)\,[x|y]^+\,VG_0\,[x|y]^+\,VG(a^2)\,[x|y]^+\\
			& \ldots \,VG(a^{k-1})\,[x|y]^+\,VG_0,
		\end{align*}
		\begin{align*}
			p''':=&[x|y]^+\,VG(a^k)\,[x|y]^+\,VG_0\,[x|y]^+\,VG(a^{k+1})\,[x|y]^+\,VG_0\\
			& \ldots \,VG_0\,[x|y]^+\,VG(a^M)\,[x|y]^+\,VG_0\,x^+
		\end{align*}
		and text $t$ as $t=t''\,t'''$, where
		\begin{align*}
			t'':=&x\,VG'(b^0)\,y^{d+10}\,VG'(b^1)\,x^{d+10}\,VG'(b^2)\,y^{d+10}\,VG'(b^3)\,x^{d+10}\,VG'(b^4)\\
				&\ldots \,VG'(b^{k-3})\,y^{d+10}\,VG'(b^{k-2})\,x^{d+10}\,VG'(b^{k-1}),
		\end{align*}
		\begin{align*}
	t''':=&y^{d+10}\,VG'(b^k)\,x^{d+10}\,VG'(b^{k+1})\,y^{d+10}\,VG'(b^{k+2})\,x^{d+10}\,VG'(b^{k+3})\,y^{d+10}\,VG'(b^{k+4})\,x^{d+10} \\
			&\ldots \,x^{d+10}\,VG'(b^M)\,y^{d+10}\,VG'(b^{M+1})\,x.
		\end{align*}
		Now we transform $p$ into $t$ in two steps - transform $p''$ into $t''$ and $p'''$ into $t'''$.

		\paragraph{Transform $p''$ into $t''$} $p''$ starts with $x^+$, which we transform into $x$. For the rest of $p''$, we make transformations according to the following rules. We make transformations starting from the beginning of $p''$. If we see $VG_0$, we transform it into the corresponding $VG'(b^j)$. If we see $[x|y]^+$, we transform it into $x^5$ or $y^5$ (according to which symbols are in the corresponding positions in $t''$). If we see $VG(a^i)$, we transform it into $x^d$ or $y^d$ (according to which symbols are in the corresponding positions in $t''$).

		\paragraph{Transform $p'''$ into $t'''$} Notice that $t'''$ starts with $y$. This is because $\mod{k}{1}{2}$.
		Now we make the following $3$ transformations from the prefix of $p'''$: transform $[x|y]^+$ into $y^{d+10}$, $VG(a^k)$ into $VG'(b^k)$ (we can do this because $a^k \cdot b^k=0$) and  $[x|y]^+$ into $x^{d+10}$. Now we transform the rest of $p'''$ into the remainder of $t'''$. We do that starting from the beginning of the remainder of $p'''$. If we see $VG_0$, we transform it into the corresponding $VG'(b^j)$. If we see $[x|y]^+$, we transform it into $x^5$ or $y^5$ (depending on which symbols are in the corresponding positions in $t$). If we see $VG(a^i)$, we transform it into $x^d$ or $y^d$ (depending on which symbols are in the corresponding positions in $t$). Finally, to finish the transformation, we transform $x^+$ into $x$. %\todo{This whole paragraph needs be rephrased}
	\end{proof}

	\begin{lemma} \label{nort3}
		If a substring of $t$ can be derived from $p$, then there are two orthogonal vectors.
	\end{lemma}
	\begin{proof}
		Recall that the pattern $p$ consists of $M+1$ vector gadgets $VG_0$. We enumerate the gadgets with integers $0, 1, 2, 3, \ldots, M$. Each of them consists of $d$ symbols from the alphabet $\{0,1,0',1'\}$. Therefore, by the construction of $t$, it must be the case that every vector gadget  $VG_0$ transforms into $VG'(b^j)$ for some $j$. Assume that the $0$-th $VG_0$ transforms into $VG'(b^0)$. If it transforms into $VG'(b^j)$ for some other $j$, $\mod{j}{0}{2}$, the proof is analogous.  The  modularity constraint on $j$ holds because $p$ starts with $x$ and the symbol $x$ must precede $VG'(b^j)$. We consider two cases below: 
		
		\paragraph{Case $1$} There exists $t\in \{0,1,2,\ldots,M\}$ such that the $t$-th $VG_0$ is not transformed into $VG'(b^t)$. Pick smallest such $t$. This means that  $VG'(b^{t})$ has been derived from  $VG(a^t)$. From the construction of $VG(a^t)$ and $VG'(b^t)$, we conclude that $a^t \cdot b^t=0$.
		
		\paragraph{Case $2$} For every $t \in \{0,1,2,\ldots,M\}$, the $t$-th $VG_0$ is transformed into $VG'(b^t)$. Consider the $M$-th vector gadget $VG_0$. It is transformed into $VG'(b^M)$. The $M$-th vector gadget $VG_0$ is followed by $x^+$ and $VG'(b^M)$ is followed by $y$. This means that we can't derive this substring of $t$ from $p$. Thus, Case $2$ cannot happen.
	\end{proof}

\subsection{Hardness for type ``${\circ}|+$''}
\label{cop}
	\begin{theorem} \label{cop_pm}
		Given sets $A=\{a^1,\ldots,a^M\} \subseteq\{0,1\}^d$ and  $B=\{b^1,\ldots,b^N\} \subseteq\{0,1\}^d$  with $M\leq N$, we can construct a regular expression $p$  and a sequence of symbols $t$,  in $O(Nd)$ time, such that  a substring of $t$ can be derived from $p$ iff there are $a \in A$ and $b \in B$ such that $a \cdot b=0$.
		Furthermore $p$ has type ``${\circ}|+$'', $|p|\leq O(Md)$ and $|t|\leq O(Nd)$.
	\end{theorem}
	\begin{proof}
		We will modify the construction for ``${\circ}{+}|$'' so that it gives the hardness proof for ``${\circ}|+$''. Whenever we have a regular expression in $p$ of the form $[s_1|s_2|s_3|\ldots|s_l]^+$ for $l\geq 1$ symbols $s_1, s_2,\ldots,s_l$, we replace it with $[s_1^+|s_2^+|s_3^+|\ldots|s_l^+]$. Let $p'$ be the new regular expression that we obtain this way. The text $t$ remains unchanged. Observe that, if we can derive some sequence $x$ from $p'$, we were able to derive $x$ from $p$ as well. Because of this, if a subsequence of $t$ can be derived from $p$, then we can conclude that there are two orthogonal vectors between $A$ and $B$. It remains to show that, if there are two orthogonal vectors, then a subsequence of $t$ can be derived from $p$. This follows from the proof of Lemma \ref{ort3}. In particular, we observe that in the proof of Lemma \ref{ort3}, if we derive a sequence from $[s_1|s_2|s_3|\ldots|s_l]^+$, such sequence is of the form $s_i^j$ for some $j\geq 1$ and $i \in [l]$. %The rest of the properties follows from the construction.
	\end{proof}
	
	\section{Reductions for the Membership problem}

\subsection{Hardness for type ``$\circ *$''}
\label{mcs}
	\begin{theorem}
		Given sets $A=\{a^1,\ldots,a^N\} \subseteq\{0,1\}^d$ and  $B=\{b^1,\ldots,b^N\} \subseteq\{0,1\}^d$, we can construct a regular expression $p$ and a sequence of symbols $t$,  in $O(Nd)$ time, such that $t$ can be derived from $p$ if and only if there are $a \in A$ and $b \in B$ such that $a \cdot b=0$.
		Furthermore, $p$ is of type ``$\circ *$'', $|p|,|t|\leq O(Nd)$.
	\end{theorem}
	\begin{proof}
		We slightly modify the construction from Theorem \ref{cs_pm}. We instantiate the construction from Theorem \ref{cs_pm} with $M=N$. We obtain a pattern $p'$ and a text $t$ such that a substring of $t$ can be derived from $p'$ iff there are two orthogonal vectors. We define the pattern $p$ as follows:
		$$
			p:=\left(\bigcirc_{j=1}^{|t|}(x^* y^*)\right) \circ p' \circ \left(\bigcirc_{j=1}^{|t|}(x^* y^*)\right).
		$$
		We claim that $t$ can be derived from $p$ iff there are two orthogonal vectors. This follows from construction of $p$ and Theorem \ref{cs_pm}. We know that a substring of $t$ can be derived from $p'$ iff there are two orthogonal vectors. Expressions $\bigcirc_{j=1}^{|t|}(x^* y^*)$ allow us to derive the remaining prefix and suffix of $t$.
	\end{proof}

\subsection{Hardness for type ``$\circ{+}\circ$''}
\label{mcpc}
	\begin{theorem} \label{cpc_memb}
		Given sets $A=\{a^1,\ldots,a^N\} \subseteq\{0,1\}^d$ and  $B=\{b^1,\ldots,b^N\} \subseteq\{0,1\}^d$, we can construct a regular expression $p$ and a sequence of symbols $t$,  in $O(Nd)$ time, such that $t$ can be derived from $p$ if and only if there are $a \in A$ and $b \in B$ such that $a \cdot b=0$.
		Furthermore, $p$ is a concatenation of ``$+$'' of sequences, $|p|,|t|\leq O(Nd)$.
	\end{theorem}
	\begin{proof}
		We will adapt the construction from Theorem \ref{cpct}. We instantiate the construction from Theorem \ref{cpct} with $M=N$. We obtain a pattern $p'$ and a text $t$ such that a substring of $t$ can be derived from $p'$ iff there are two orthogonal vectors. The final pattern $p:=p_1 \circ p' \circ p_2$ is a concatenation of three expressions $p_1, p', p_2$. Each one of expressions $p_1, p'$ and $p_2$ is a concatenation of ``$+$'' of sequences. Clearly, if $t$ can be derived from $p$, then a substring of $t$ can be derived from $p'$. By the statement of Theorem \ref{cpct}, there must be two orthogonal vectors in this case. Therefore our goal is to construct $p_1$ and $p_2$ such that text $t$ can be derived from $p=p_1 \circ p' \circ p_2$ if there are two orthogonal vectors. In the rest of this section we achieve this goal.
		
		If there are two orthogonal vectors, then by the proof of Theorem \ref{cpct}, the text 
		$$
		t=\bigcirc_{j=-5N}^{6N}\left(VG'(0_d)\,VG_0'\,VG'(b^j)\,VG'(0_d)\,VG_0'\,VG_1'\right)
		$$ 
		can be written as $t=t_1 t_2 t' t_3 t_4$, where the sequences $t_1, t_2, t', t_3$ and $t_4$ have the following properties.
		\begin{itemize}
			\item
				$$
					t_1=\bigcirc_{j=-5N}^{w-N}\left(VG'(0_d)\,VG_0'\,VG'(1_d)\,VG'(0_d)\,VG_0'\,VG_1'\right).
				$$
			\item
				\begin{align*}
					t_2=& \left(\bigcirc_{j=w-N+1}^{w-1}\left(VG'(0_d)\,VG_0'\,VG'(b^j)\,VG'(0_d)\,VG_0'\,VG_1'\right)\right) \\ 
						& \circ VG'(0_d)\,VG_0'\,VG'(b^w)\,VG'(0_d)\,VG_0'.
				\end{align*}
			\item
				$$
					t'=VG_1'\,\bigcirc_{j=1+w}^{N+w}\left(VG'(0_d)\,VG_0'\,VG'(b^j)\,VG'(0_d)\,VG_0'\,VG_1'\right)
				$$
				for some $w \in \{1-N,\ldots,N-1\}$ and $t'$ can be derived from $p'$.
			\item
				$$
					t_3=\bigcirc_{j=N+w+1}^{w+2N-1}\left(VG'(0_d)\,VG_0'\,VG'(b^j)\,VG'(0_d)\,VG_0'\,VG_1'\right).
				$$
			\item
				$$
					t_4= \bigcirc_{j=w+2N}^{6N}\left(VG'(0_d)\,VG_0'\,VG'(1_d)\,VG'(0_d)\,VG_0'\,VG_1'\right).
				$$
		\end{itemize}
		Our goal is to construct expressions $p_1$ and $p_2$ such that $t_1t_2$ can be derived from $p_1$ (independently of the value $w$) and $t_3t_4$ can be derived from $p_2$ (independently of the value $w$). We construct $p_1$ and $p_2$ as follows.
		\begin{itemize}
			\item We note that the expression
				$$
					\hat p:=\bigcirc_{j=1}^{N}\left(\left[VG'(0_d)\,VG_0'\,VG'(1_d)\,VG'(0_d)\,VG_0'\,VG_1'\right]^{+}\right)
				$$
				can derive the sequence $t_1$ independently of the value $w$. Let $z$ be an arbitrary sequence of symbols $x$ and $y$. As long as there are two neighboring symbols $x$ ($y$, resp.), we replace those two symbols by one copy of symbol $x$ ($y$, resp.). Let $z'$ be the resulting sequence. We define expression $f(z)$ as follows:
				$$
					f(z):=\bigcirc_{j=1}^{|z'|}[z'_j]^{+}.
				$$
				That is, in the expressions $f(z)$ we allow to repeat any symbol in $z'$ one or more times.
				Note that $f(t_2)$ can derive $t_2$ independently of the value $w$ and that $f(t_2)$ does not depend on vectors $b^j$. This implies that the expression $p_1:=\hat p \circ f(t_2)$ can derive $t_1t_2$ which is what we needed.
			\item We define $p_2:=f(t_3) \circ \hat p$. We can check that $f(t_3)$ and $\hat p$ do not depend on the value $w$ and vectors $b^j$. We can also check that we can derive $t_3 t_4$ from $p_2$.
		\end{itemize}
	\end{proof}

\subsection{Hardness for type ``$\circ|\circ$''}
\label{mcoc}
	\begin{theorem}
		Given sets $A=\{a^1,\ldots,a^N\} \subseteq\{0,1\}^d$ and  $B=\{b^1,\ldots,b^N\} \subseteq\{0,1\}^d$, we can construct a regular expression $p''$ and a sequence of symbols $t$,  in $O(Nd)$ time, such that $t$ can be derived from $p''$ if and only if there are $a \in A$ and $b \in B$ such that $a \cdot b=0$.
		Furthermore, $p''$ is of type ``$\circ|\circ$'', $|p''|,|t|\leq O(Nd)$.
	\end{theorem}
	\begin{proof}
		We will modify the construction for ``$\circ{+}\circ$'' (Theorem \ref{cpc_memb}) so that it gives a hardness proof for ``$\circ|\circ$''. The text $t$ remains the same. We will modify pattern $p$ as follows. First, recall that $p=p_1 \circ p' \circ p_2$.
		We transform $p$ into an expression of type ``$\circ|\circ$'' in two steps.
		\begin{itemize}
			\item We transform $p'$ into a sequence of type ``$\circ|\circ$'' in the same way as it is done in the proof of Theorem \ref{coc_pm}.
			\item Expressions $p_1$ and $p_2$ are repeated concatenations of expressions 
				$$
					[x]^{+}, \,\, [y]^{+}, \,\, \left[VG'(0_d)\,VG_0'\,VG'(1_d)\,VG'(0_d)\,VG_0'\,VG_1'\right]^{+}.
				$$ 
				In the proof of Theorem \ref{cpc_memb} we can repeat each one of argument expressions 
				$$
					x, \,\, y, \,\, VG'(0_d)\,VG_0'\,VG'(1_d)\,VG'(0_d)\,VG_0'\,VG_1'
				$$ 
				at most $8$ times so that we are still able to derive sequences $t_1t_2$ and $t_3t_4$ from $p_1$ and $_2$ respectively. Thus, we replace $[s]^{+}$ by $[s\,|\,s^2\,|\,s^3\,|\,s^4\,|\,s^5\,|\,s^6\,|\,s^7\,|\,s^8]$ for each
				$$
					s \, = \, x, \,\, y, \,\, VG'(0_d)\,VG_0'\,VG'(1_d)\,VG'(0_d)\,VG_0'\,VG_1'
				$$ 
				in $p_1$ and $p_2$.
		\end{itemize}
		Let $p''$ be the resulting expression.
		If the sequence $t$ can be derived from $p$, it can still be derived from $p''$. This follows from Theorem \ref{coc_pm} and the construction of $t$. It remains to argue that if $t$ can't be derived from $p$, then it can't be derived from $p''$. This is true by the transformation above and Theorem \ref{coc_pm}.
	\end{proof}

\subsection{Hardness for type ``$\circ {+}|$''}
\label{mcpo}
	\begin{theorem}
		Given sets $A=\{a^1,\ldots,a^M\} \subseteq\{0,1\}^d$ and  $B=\{b^1,\ldots,b^N\} \subseteq\{0,1\}^d$  with $M\leq N$, we can construct a regular expression $p$  and a sequence of symbols $t$,  in $O(Nd)$ time, such that $t$ can be derived from $p$ iff there are $a \in A$ and $b \in B$ such that $a \cdot b=0$.
		Furthermore $p$ has type ``$\circ{+}|$'', $|p|\leq O(Md)$ and $|t|\leq O(Nd)$.
	\end{theorem}
	\begin{proof}
		We adapt the hardness proof from Theorem \ref{cpo_theorem}. We instantiate the construction from Theorem \ref{cpo_theorem} and we obtain a pattern $p'$ and a text $t$ such that a substring of $t$ can be derived from $p'$ iff there are two orthogonal vectors.
		We define the new pattern $p$ as follows:
		$$
			p:=[0 \, | \, 1 \, | \, 0' \, | \, 1' \, | \, x \, | \, y]^+ \, \circ \, p' \, \circ \, [0 \, | \, 1 \, | \, 0' \, | \, 1' \, | \, x \, | \, y]^+.
		$$
		We claim that $t$ can be derived from $p$ iff there are two orthogonal vectors. If $t$ can be derived from $p$, then a substring of $t$ can be derived from $p'$ and by Theorem \ref{cpo_theorem} there are two orthogonal vectors. Conversely, if there are two orthogonal vectors then by Theorem \ref{cpo_theorem} we can derive a substring of $t$ from $p'$. We derive the remaining prefix and suffix of $t$ from expressions $[0 \, | \, 1 \, | \, 0' \, | \, 1' \, | \, x \, | \, y]^+$.
	\end{proof}

\subsection{Hardness for type ``$\circ|+$''}
\label{mcop}
	\begin{theorem}
		Given sets $A=\{a^1,\ldots,a^N\} \subseteq\{0,1\}^d$ and  $B=\{b^1,\ldots,b^N\} \subseteq\{0,1\}^d$, we can construct a regular expression $p$  and a sequence of symbols $t$,  in $O(Nd)$ time, such that $t$ can be derived from $p$ iff there are $a \in A$ and $b \in B$ such that $a \cdot b=0$.
		Furthermore $p$ has type ``$\circ|+$'', $|p|,|t|\leq O(Nd)$.
	\end{theorem}
	\begin{proof}
		We adapt the hardness proof from Theorem \ref{cop_pm}. We instantiate the construction from Theorem \ref{cop_pm} with $M=N$. We obtain a pattern $p'$ and a text $t$ such that a substring of $t$ can be derived from $p'$ iff there are two orthogonal vectors. From the proof of Theorem \ref{cpo_theorem} we have that
		$$
			t=\bigcirc_{j=-9N}^{10N}\left(x^{d+10}\,VG'(b^{2j})\,y^{d+10}\,VG'(b^{2j+1})\right).
		$$
		\iffalse
		We call a substring of a sequence \emph{nice} if there are no two symbols in the substring that are different and the substring is maximal, i.e., we can't add a symbol to the left or to the right while maintaining that the substring is nice.
		Let $k$ denote the number of substrings of $x^{d+10}\,VG'(b^{2j})\,y^{d+10}\,VG'(b^{2j+1})$ that are nice. Notice that $k$ does not depend on $j$.
		\fi
		Let $k=|x^{d+10}\,VG'(b^{2j})\,y^{d+10}\,VG'(b^{2j+1})|$ be the length of sequence $x^{d+10}\,VG'(b^{2j})\,y^{d+10}\,VG'(b^{2j+1})$. Notice that $k$ does not depend on $j$.
		Our new sequence $p$ is constructed as follows:
		$$
			p:=\left(\bigcirc_{j=1}^{7Nk}[0^+ \, | \, 1^+ \, | \, [0']^+ \, | \, [1']^+ \, | \, x^+ \, | \, y^+]\right) \, \circ \, p' \, \circ \, \bigcirc_{j=1}^{7Nk}[0^+ \, | \, 1^+ \, | \, [0']^+ \, | \, [1']^+ \, | \, x^+ \, | \, y^+].
		$$
		
		If $t$ can be derived from $p$, then a substring of $t$ can be derived from $p'$ and by Theorem \ref{cop_pm} there are two orthogonal vectors. If there are two orthogonal vectors then we can derive a substring of $t$ from $p'$. We can easily check that we can derive the remaining prefix and suffix of $t$ from expressions $\bigcirc_{j=1}^{7Nk}[0^+ \, | \, 1^+ \, | \, [0']^+ \, | \, [1']^+ \, | \, x^+ \, | \, y^+]$.
	\end{proof}
	
	\section{Algorithms}

	\subsection{Algorithm for the Word Break problem}
\label{wb}
\paragraph{Word Break problem}
	Given a binary sequence $t$ of length $|t|=n$ and a collection of binary\footnote{W.l.o.g. we assume that all sequences are binary. If this is not so, we encode every symbol of the alphabet using a binary sequence of length $\lceil \log s \rceil +1$ where $s$ is the size of the alphabet. This increases the lengths of the sequences by a logarithmic multiplicative factor.} sequences $p$
	with total length $\sum_{p' \in p}|p'|=m$,
	decide if the sequence $t$ can be written as a concatenation $t=t_1 \ldots t_k$ such that $t_i \in p$ for every $i \in [k]$.
	If $t$ can be written in such a way, we call $t$ \emph{decomposable}.

We will solve this problem in time $\tO(n \cdot m^{0.4444\ldots})=\tO\left(n \cdot m^{0.5-\nicefrac{1}{18}}\right)$.\footnote{$\tO(\cdot)$ notation hides a $\poly \log$ factor.}

As a warm-up, we first solve the problem in time $\tO(n \cdot \sqrt m)$ and then we provide a $\tO\left(n \cdot m^{0.5-\nicefrac{1}{18}}\right)$ time algorithm.

\subsubsection{$\tO(n \cdot \sqrt m)$ time algorithm for the Word Break problem}
	Let $d(p):=\{|p'| \ : \ p' \in p\}$. We will show how to solve the Word Break problem in time $\tO(n \cdot |d(p)|)$.
	Since $\sum_{i \in [t]}|p_i|=m$, we have that $|d(p)|\leq O(\sqrt m)$, which implies the upper bound.
	
	We will use the following lemma.
	\begin{lemma} \label{hash}
		We can randomly choose a hash function $h: \{0,1\}^* \to \N$ and
		preprocess $t$ in $\tO(n)$ time such that the following holds:
		\begin{itemize}
			\item Given any substring $t'$ of $t$, we can compute the hash $h(t')$ in $\tO(1)$ time.
			\item For any two sequences $t''\neq t'$ (not necessarily substrings of $t$), $\Pr[h(t'')=h(t')]\leq 1/n^{10}$.
			\item For any sequence $t'$ (necessarily substring of $t$), we can compute $h(t')$ in time $\tO(|t'|)$.
		\end{itemize}
	\end{lemma}
	\begin{proof}
		E.g., use Rabin-Karp rolling hash.%\todo{Show more details?}
	\end{proof}
	
	\begin{theorem} \label{warmup}
		The Word Break problem can be solved in time $\tO(n \cdot |d(p)|)$.
	\end{theorem}
	\begin{proof}
		Preprocess $t$ according to Lemma \ref{hash} and compute $h(p):=\{h(p') \ : \ p' \in p\}$.
		
		We solve the problem using dynamic programming. We use the table $D:[n+1] \to \{0,1\}$. The algorithm determines the values $D[n], D[n-1], \ldots, D[1]$ (in this order). We set $D[i]=1$ iff the sequence $t_it_{i+1}\ldots t_n$ is decomposable.
		
		\begin{itemize}
			\item Set $D[i]=0$ for all $i=1,\ldots,n$ and set $D[n+1]=1$.
			\item For $i=n,\ldots,1$, set $D[i]=1$ iff there exists $j>i$ such that $D[j]=1$ and $(j-i) \in d(p)$, and $h(t_i\ldots t_{j-1}) \in h(p)$.
			\item $t$ is decomposable iff $D[1]=1$.
		\end{itemize}
	\end{proof}

\subsubsection{$\tO\left(n \cdot m^{0.5-\nicefrac{1}{18}}\right)$ time algorithm for the Word Break problem}
	Similarly as in the $\tO(n \cdot \sqrt m)$ time algorithm, we fill out the table $D:[n+1] \to \{0,1\}$:
	\begin{itemize}
		\item Set $D[i]=0$ for all $i=1,\ldots,n$ and set $D[n+1]=1$.
		\item For every $i=n,\ldots,1$ in this order, set $D[i]=1$ iff the sequence $t_it_{i+1}\ldots t_n$ is decomposable.
		\item $t$ is decomposable iff $D[1]=1$.
	\end{itemize}
	We will show that the second step can be performed in $\tO(n \cdot m^{0.5-\alpha})$ time for sufficiently small constant $\alpha>0$. We will later show that we can set $\alpha=\nicefrac{1}{18}$. For now we can think of $\alpha>0$ as a sufficiently small constant, say, $\alpha=0.01$. We will make the following two assumptions, justified by the next two lemmas.
	\begin{lemma} \label{length_lb}
		For all $p' \in p$, $|p'|\geq m^{0.5-\alpha}$.
	\end{lemma}
	\begin{proof}
		Let $\hat p:=\{p' \in p \ : \ |p'|<m^{0.5-\alpha}\}$.
		Clearly, we have that $|d(\hat p)|<m^{0.5-\alpha}$. Therefore, as we perform the second step of the algorithm, for every $i=n,\ldots,1$, we set $D[i]=1$ if there exists $j>i$ with $D[j]=1$ and $t_i\ldots t_{j-1} \in \hat p$ in the same way as it is done in the proof of Theorem \ref{warmup}. For every $i$ this takes $\tO(|d(\hat p)|) = \tO(m^{0.5-\alpha})$ time. Therefore, the total runtime corresponding to processing sequences in $\hat p$ is $\tO(n \cdot m^{0.5-\alpha})$.
	\end{proof}
	
	\begin{lemma} \label{length_ub}
		For all $p' \in p$, $|p'|\leq m^{0.5+\alpha}$.
	\end{lemma}
	\begin{proof}
		Let $\hat p:=\{p' \in p \ : \ |p'|>m^{0.5+\alpha}\}$. Since $\sum_{p' \in p}|p'|=m$, we have $|d(\hat p)|\leq |\hat p|\leq m^{0.5-\alpha}$. Similarly as in the proof of Lemma \ref{length_lb} we can set $D[i]=1$ if there exists $j>i$ with $D[j]=1$ and $t_i\ldots t_{j-1} \in \hat p$. Therefore, the total runtime corresponding to processing sequences in $\hat p$ is $\tO(n \cdot |\hat p|)\leq \tO(n \cdot m^{0.5-\alpha})$.
	\end{proof}
	
	In the rest of the section we will show that the second step of the algorithm can be implemented in $\tO(n \cdot m^{0.5-\alpha})$ time if $\alpha>0$ is a sufficiently small constant. By Lemmas \ref{length_lb} and \ref{length_ub}, we can assume that for all $p' \in p$, $m^{0.5-\alpha}\leq |p'| \leq m^{0.5+\alpha}$.

	We build a trie data structure $T$ for $p$. It is a binary tree where each node has two children. Each node corresponds to a prefix of a sequence in $p$. The root node corresponds to an empty sequence. If a node $u$ corresponds to the sequence $s$ and has two children then one of the children corresponds to the sequence $s0$ ($s$ followed by $0$) and the other corresponds to the sequence $s1$. If $u$ has only one child, it corresponds to either $s0$ or $s1$. If a node $u$ corresponds to a sequence of length $i\geq 0$, $u$ has \emph{depth} $i$. If a node $u$ corresponds to a sequence $p'$ and $p' \in p$, then we call the node $u$ \emph{marked}. We preprocess $T$ in such a way that for any node $u$ we have a pointer to node $v$ such that $v$ is a marked ancestor of $u$ of maximal depth. A node is not its own ancestor. This data structure can be constructed in $O(m)$ time. Because $\sum_{p' \in p}|p'|=m$ and for all $p' \in p$, $m^{0.5-\alpha}\leq |p'|$, we have the the number of marked nodes in the tree is upper bounded by $|p|\leq m^{0.5+\alpha}$.

	\paragraph{Preprocessing of the tree $T$} We further preprocess the tree $T$ and the sequence $t$ such that given any index $i=1,\ldots, n$, we can answer the following query in $\tO(1)$ time. Specifically, the query algorithm outputs the maximal $j>i$ such that $t_it_{i+1}\ldots t_{j-1}$ is a prefix of a sequence in $p$, and reports the node $u$ corresponding to $t_it_{i+1}\ldots t_{j-1}$. We build a data structure that stores the hash values for all non-empty prefixes of all sequences in $p$ and supports lookups in $\tO(1)$ time. We use Rabin-Karp rolling hash to compute the hashes. This takes $\tO(m)$ time. We also preprocess the sequence $t$ so that we can compute Rabin-Karp rolling hash value in $\tO(1)$ time for every substring. This takes $\tO(n)$. The total runtime of the preprocessing steps is $\tO(m+n)$. Given $i$, the query algorithm does the binary search to find the largest $j>i$ such that the hash of sequence $t_it_{i+1}\ldots t_{j-1}$ is in the table. Since we can do lookups in the table in $\tO(1)$ and there are $\tO(1)$ binary search steps, and we can compute the rolling hash value for every substring in $\tO(1)$ time, we get the required upper bound $\tO(1)$ on the query time. We can easily augment the data structure so that we can output the corresponding node $u$ from the tree $T$.
	%\todo{Implement this preprocessing using the suffix trees}

	For $j=1,\ldots, n/m^{0.5-\alpha}$ we call the sequence $D[n-j\cdot m^{0.5-\alpha}+1]\ldots D[n-(j-1)\cdot m^{0.5-\alpha}]$ the $j$-th {\em chunk} of the table $D$.
	We will show how to determine the values in the $j$-th chunk in time $\tO(m^{1-2\alpha})$ assuming that we have all values for chunks with indices smaller than $j$. This gives the required upper bound on the runtime because there are $n/m^{0.5-\alpha}$ chunks and $\tO(m^{1-2\alpha})\cdot n/m^{0.5-\alpha} = \tO(n \cdot m^{0.5-\alpha})$.
	Since for all $p' \in p$, $m^{0.5-\alpha}\leq |p'| \leq m^{0.5+\alpha}$, the only previously computed values of $D$ that are needed  to compute the $j$-th chunk are:
	$$
		D[n-(j-1)\cdot m^{0.5-\alpha}+1]\ldots D[n-(j-1)\cdot m^{0.5-\alpha}+m^{0.5+\alpha}]
	$$
	 To simplify the notation, we relabel the table $D$ by defining the table $D'$.
	Specifically, we identify $D[n-j\cdot m^{0.5-\alpha}+1]\ldots D[n-(j-1)\cdot m^{0.5-\alpha}+m^{0.5+\alpha}]$ with $D'[1]\ldots D'[m^{0.5-\alpha}+m^{0.5+\alpha}]$.  Thus, our goal is to find the values $D'[1]\ldots D'[m^{0.5-\alpha}]$ knowing the values $D'[m^{0.5-\alpha}+1]\ldots D'[m^{0.5-\alpha}+m^{0.5+\alpha}]$. Let $t'$ be the corresponding substring of $t$ of length $m^{0.5-\alpha}+m^{0.5+\alpha}$.
	
	\paragraph{Intuition} Suppose that we want to determine value of $D'[i]$. We look for the largest $j>i$ such that $t'_i\ldots t'_{j-1}$ is in $p$. If $D'[j]=1$, we set $D'[i]=1$ and move to determine $D'[i-1]$. However, it might be that $D'[j]=0$ and there are integers $j'$ such that $i<j'<j$ and $t'_i \ldots t'_{j'-1}$ is in $p$. For every such $j'$ we have to check whether $D'[j']=1$ and set $D'[i]=1$ if this happens. If there are not too many such $j'$, we can work through all of them. It might happen that there are many such $j'$. In this case we build a characteristic vector of the set of such $j'$s, i.e.,  set the entry corresponding to each such $j'$ to $1$. We then convolve the characteristic vector with $D'$. Although this does not reduce the runtime when working with $D'[i]$, the saving will occur in the future if we will need to determine $D'[i']$ such that $i'<i$ and $t'_i t'_{i+1} \ldots$ and $t'_{i'}t'_{i'+1}\ldots$ share long prefixes. We will make this more precise below.
	
	We determine the unknown values of $D'$ in two phases - preprocessing phase and online phase. In the preprocessing phase we preprocess the known part of $D'$ together with $T$ in $\tO(m^{1-2\alpha})$ time. In the online phase we determine the unknown values $D'[i]$ for $i=m^{0.5-\alpha},\ldots, 1$ in this order. We spend time $O(m^{0.5-\alpha})$ for every $D'[i]$.
	
	\paragraph{Preprocessing phase}
		In the following we will define a subset of the marked nodes that we call \emph{special}.  Initially the set of special nodes is empty.
		We will keep the invariant that if a node is special, then all its marked ancestors are also special.
		Since $|p|\leq m^{0.5+\alpha}$, there are at most $m^{0.5+\alpha}$ leaves in the tree $T$. Fix an arbitrary ordering of the leaves and consider the leaves one by one. Let $\ell$ be the current leaf that we consider. Let $c$ be the number of marked ancestors of $\ell$ that are not special. We can determine $c$ in time $O(c)$ because every node in $T$ keeps a pointer to the marked ancestor of maximal depth. We distinguish two cases.

		\paragraph{Case 1: $c>m^{0.5-3\alpha}$} We mark $\ell$ and all its marked ancestors as special. Since $|p|\leq m^{0.5+\alpha}$, the number of marked nodes is at most $m^{0.5+\alpha}$. This means that we happen to be in this case at most $m^{0.5+\alpha}/m^{0.5-3\alpha}=m^{4\alpha}$ times. Let $d$ denote the depth of the node $\ell$. Since for every $p' \in p$, $|p'|\leq m^{0.5+\alpha}$, we have that $d \leq m^{0.5+\alpha}$. Let $u_0, u_1, \ldots, u_d=\ell$ be the nodes on the path from the root of $T$ to $\ell$. The node $u_0$ is the root of $T$. Let $l:=m^{0.5-\alpha}$. We define $\lfloor d/l \rfloor$ binary vectors $r_1, \ldots, r_{\lfloor d/l \rfloor} \in \{0,1\}^l$ as follows. For $i=1,\ldots,\lfloor d/l \rfloor$ and $j=1,\ldots,l$ we set the $j$-th entry $r_i[j]$ of the vector $r_i$ to be equal to $r_i[j]=1$ if the node $u_{(i-1)l+j}$ is marked and equal to $r_i[j]=0$ if the node is not marked. For every $i=1,\ldots,\lfloor d/l \rfloor$, we compute the convolution between the binary vector $r_i$ and the binary vector $D'[m^{0.5-\alpha}+1]\ldots D'[m^{0.5-\alpha}+m^{0.5+\alpha}]$ in the following sense. We output a binary vector $c_i$ with $m^{0.5-\alpha}+m^{0.5+\alpha}+l-1$ entries such that for $j=1,\ldots,m^{0.5-\alpha}+m^{0.5+\alpha}+l-1$, 
		$$
			c_i[j]:=\sum_{k=1}^l\left( r_i[k] \cdot D'[k+j+m^{0.5-\alpha}-l] \right).
		$$
		When computing $c_i[j]$, if we need to access an entry $D'[z]$ with $z\leq m^{0.5-\alpha}$ or $z \geq m^{0.5-\alpha}+m^{0.5+\alpha}+1$, we assume that it is equal to $0$.
		Computing the convolution $c_i$ for $r_i$ takes $O(m^{0.5+\alpha}\log m)$ time using the Fast Fourier Transform. Since we have to compute the convolution for $i=1,\ldots,\lfloor d/l \rfloor$, in total it takes $O(m^{0.5+\alpha}\log m) \cdot d/l=O(m^{0.5+3\alpha} \log m)$ time. Since we happen to be in this case at most $m^{4\alpha}$ times, the total runtime corresponding to this case is bounded by $O(m^{0.5+7\alpha} \log m) = \tO(m^{1-2\alpha})$ assuming that $\alpha\leq 1/18$.

		\paragraph{Case 2: $c\leq m^{0.5-3\alpha}$} In this case we do not do anything. Since there are at most $m^{0.5+\alpha}$ leaves, the total time corresponding to this case is upper bounded by $m^{0.5+\alpha} \cdot O(c) \leq O(m^{1-2\alpha})$ which is what we wanted.
		
		\begin{figure*}
        \centering
		
		\begin{tikzpicture} [solid]
			\node at (-2.5,0) {Tree $T$:};
		
			\draw (-5,-10) -- (0,0) -- (5,-10);
			
			\draw[dashed] (-6,-5) -- (6,-5);
			\draw[dashed] (-6,-9.5) -- (6,-9.5);
			
			\node at (7.5,-5) {Depth $m^{0.5-\alpha}$};
			\node at (7.5,-9.5) {Depth $m^{0.5+\alpha}$};
			
			\draw (0,0) .. controls (0,-4) and (-3,-5) .. (-2,-8.5);
			
			\draw (-2,-8.5) circle (0.7mm);
			\draw (-2.15,-7.9) circle (0.7mm);
			\draw (-2.2,-7.3) circle (0.7mm);
			\draw (-2.13,-6.3) circle (0.7mm);
			\draw (-2.03,-5.9) circle (0.7mm);
			\draw (-1.8,-5.3) circle (0.7mm);
			
			\node[anchor=east] at (-2,-8.5) {$\ell_1$};

			\draw (0,0) .. controls (1,-4) and (-1,-5) .. (0,-9);
			
			\draw[fill=black] (0,-9) circle (0.7mm);
			\draw[fill=black] (-0.05,-8.8) circle (0.7mm);
			\draw[fill=black] (-0.14,-8.3) circle (0.7mm);
			\draw[fill=black] (-0.23,-7.7) circle (0.7mm);
			\draw[fill=black] (-0.25,-7.5) circle (0.7mm);
			\draw[fill=black] (-0.28,-7) circle (0.7mm);
			\draw[fill=black] (-0.285,-6.8) circle (0.7mm);
			\draw[fill=black] (-0.28,-6.6) circle (0.7mm);
			\draw[fill=black] (-0.25,-6) circle (0.7mm);
			\draw[fill=black] (-0.17,-5.5) circle (0.7mm);
			
			\node[anchor=north west] at (-0,-9) {$\ell_2$};

			\draw (0.26,-2.5) .. controls (2,-3) and (2,-8) .. (2.8,-8.3);
			
			\draw (2.8,-8.3) circle (0.7mm);
			\draw (2.425,-7.7) circle (0.7mm);
			\draw (2.225,-7) circle (0.7mm);
			\draw (2.11,-6.5) circle (0.7mm);
			\draw (1.89,-5.5) circle (0.7mm);
			
			\node[anchor=north] at (2.8,-8.3) {$\ell_4$};

			\draw (-0.25,-6.2) .. controls (1,-7) and (1.5,-9) .. (2,-9.3);
			
			\draw[fill=black] (2,-9.3) circle (0.7mm);
			\draw[fill=black] (1.68,-8.9) circle (0.7mm);
			\draw[fill=black] (1.45,-8.5) circle (0.7mm);
			\draw[fill=black] (1.3,-8.2) circle (0.7mm);
			\draw[fill=black] (1.15,-7.9) circle (0.7mm);
			\draw[fill=black] (0.73,-7.2) circle (0.7mm);
			\draw[fill=black] (0.25,-6.6) circle (0.7mm);
			
			\node[anchor=west] at (2,-9.3) {$\ell_3$};
		\end{tikzpicture}
		
		\captionsetup{singlelinecheck=off}
        \caption[]{
			An example of the preprocessing phase of tree $T$ and vector $D'$. Each white circle denotes a marked node that is not special. Each black circle denotes a marked node that is special. We do not depict the nodes that are not marked. Each solid line is a sequence of marked or unmarked nodes. Notice that all marked nodes are at depth at least $m^{0.5-\alpha}$ and at most $m^{0.5+\alpha}$. The tree has $4$ leaves $\ell_1, \ell_2, \ell_3, \ell_4$. We process them in order $\ell_1, \ell_2, \ell_3, \ell_4$:
			\begin{itemize}
				\item Leaf $\ell_1$ has $c=5$ marked ancestors. We assume that $m^{0.5-3\alpha}=5$ and therefore we happen to be in case $c\leq m^{0.5-3\alpha}$. We do not make the marked ancestors special.
				\item Leaf $\ell_2$ has $c=9$ marked ancestors. We are in $c>m^{0.5-3\alpha}$ case and we make $\ell_2$ and all marked ancestors of $\ell_2$ special. We split the path from the root to $\ell_2$ into shorter paths of length $l$. For every shorter path we construct a binary characteristic vector where we set entry to be equal to $1$ if the corresponding node is marked. Then we convolve each characteristic vector with binary vector $D'$.
				\item Leaf $\ell_3$ has $c=6$ marked ancestors that are not special and two marked ancestors that are special. We are in $c=6>m^{0.5-3\alpha}$ case and we make $\ell_3$ and all marked ancestors of $\ell_3$ special. Then we split the path from the root to $\ell_3$ into shorter paths of length $l$ and proceed similarly as when processing leaf $\ell_2$.
				\item Leaf $\ell_4$ has $c=4$ marked ancestors. We are in $c\leq m^{0.5-3\alpha}$ case and we do not make any of the marked ancestors special.
			\end{itemize}
		}
	\label{figure_trie}
\end{figure*}
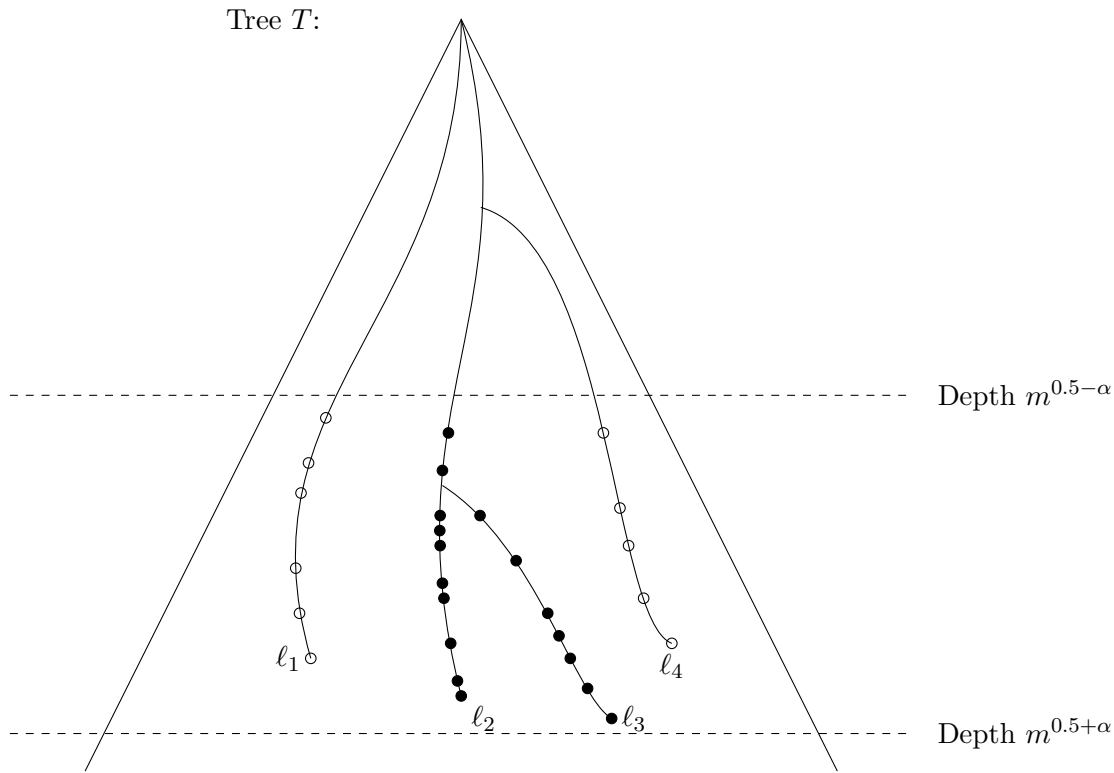

		See Figure \ref{figure_trie} for an example run of the preprocessing phase.
	
		Notice that after the preprocessing phase, the set of marked and special nodes of $T$ form a rooted subtree among all marked nodes of $T$.
	
	\paragraph{Online phase}
		We determine the values of $D'[i]$ for $i=m^{0.5-\alpha},\ldots, 1$ in this order. Fix $i$ for which we want to determine $D'[i]$. Let $j>i$ be the largest integer such that $t'_i\ldots t'_{j-1}$ corresponds to a node $u$ in tree $T$. We find $j$ using the query algorithm described before, in $\tO(1)$ time. Let $a$ be the ancestor of $u$ which is marked and special and whose depth $d$ is the largest. Let $c$ be the number of marked ancestors of $u$ that are marked but not special. By the preprocessing phase, $c\leq m^{0.5-3\alpha}$. For every such marked ancestor which is not special we want to determine whether the corresponding entry of $D'$ is equal to $1$. This takes at most $O(m^{0.5-3\alpha})$ total time. If we found such an entry of $D'$ equal to $1$, we set $D'[i]=1$ and move to determining the value of $D'[i-1]$. By the choice of the node $a$, there are at most $l$ marked ancestors of $u$ that are special and are of depth more than $d':=\lceil d/l \rceil l$. We can determine whether the corresponding entry of $D'$ is equal to $1$ for any of those nodes and set $D'[i]=1$ if this happens. This takes at most $O(l)$ total time. It remains to consider the marked ancestors of $u$ of depth at most $d'$. For this we use the convolutions that we performed in the preprocessing step. Let $u_0, u_1, \ldots, u_{d'}=a$ be the nodes on the path from the root of $T$ to $l$. The node $u_0$ is the root of $T$. Let $r_1, \ldots, r_{d'/l}$ be the binary vectors of length $l$ constructed as follows. For $k=1,\ldots, d'/l$ and $j=1,\ldots,l$ we set the $j$-th entry $r_k[j]$ of the vector $r_k$ to be equal to $r_k[j]=1$ if the node $u_{(k-1)l+j}$ is marked and equal to $r_k[j]=0$ if the node is not marked. We want to determine whether there is a marked node corresponding to an entry equal to $1$ in $r_k$ such that the corresponding entry in $D'$ is equal to $1$. We can determine whether this is the case in $O(1)$ because this we can check whether the corresponding entry of $c_k$ is equal to $0$ or at least $1$. If the entry is at least $1$, we set $D'[i]$ to be equal to $1$ and continue with $D'[i-1]$. We spend $O(1)$ for each $k=1,\ldots, d'/l$. The whole process takes $O(d'/l)$ time.
		The total runtime spent on computing the value of $D'[i]$ is bounded by $O(m^{0.5-3\alpha})+O(l)+O(d'/l)\leq O(m^{0.5-\alpha})$.
		
	\paragraph{The runtime of the algorithm} We have seen that the runtime of the algorithm is  bounded by $\tO(n \cdot m^{0.5-\alpha})$ for any constant $0<\alpha\leq \nicefrac{1}{18}$. We obtain the required runtime $\tO\left(n \cdot m^{0.5-\nicefrac{1}{18}}\right)$ by setting $\alpha=1/18$.

	\subsection{Algorithm for type ``$\circ+$''}
\label{cp}
	\begin{theorem}
		Let $p$ be a regular expression of type ``$\circ +$''  and $t$ be a text. In time $O(|p|+|t|)$ we can reduce the pattern matching problem on $p$ and $t$ to one instance of Subset Matching and one instance of Wildcard Matching.
	\end{theorem}
	\begin{proof} Below we reduce this regular expression pattern matching problem to Subset Matching and to Wildcard Matching. Then we combine the two outputs and solve the initial regexp problem.

\paragraph{Reduction to Subset Matching}
		We partition the pattern $p$ into substrings of maximal length such that all symbols in each substring are equal, i.e., each substrings is a concatenation of copies of $a^+$ or $a$. 
		% that is, the substrings are of the form $a^+a^+\ldots a^+$ for some symbol $a\in \Sigma$. 
		Consider one particular substring and suppose that it contains $l$ copies of $a$ or $a^+$. We replace this substring with a set      $\{(a,1),(a,2), \ldots, (a,l)\}$. We perform this operation for every substring of maximal length and concatenate the resulting sets to obtain the pattern $p'$ for Subset Matching. Similarly, we partition the text $t$ into substrings of maximal length such that all symbols in each substring are equal. As before, every such substring is replaced with a set $\{(a,1),(a,2), \ldots, (a,l)\}$. This yields a text $t'$.

		Now we run Subset Matching algorithm on the pattern $p'$ and the text $t'$. For each position $i$ of $t'$, we find whether $p'$ matches  the substring of $t'$ that starts at $i$.

\paragraph{Reduction to Wildcard Matching} Similarly as before, we partition pattern $p$ into substrings of maximal length such that all symbols in each substring are equal, i.e., each substrings is a concatenation of copies of $a^+$ or $a$. There are two kinds of substrings. If the substring contains only $l$ symbols $a$ (i.e., there are no $a^+$s), we replace this substring with the symbol $(a,l)$. If the substring has at least one $a^+$, we replace it with a wildcard. We concatenate all symbols and wildcards into one regular expression $p''$. If $p''$ starts with a symbol, we replace it with a wildcard. Also, if $p''$ ends with a symbol, we replace it with a wildcard.
Similarly, we partition the text $t$ into substrings of maximal length such that all symbols in each substring are equal. We replace each such substring, say $a^l$, with a symbol $(a,l)$. The symbols are concatenated into the text $t''$.

		Now we run Wildcard Matching algorithm on the pattern $p''$ and the text $t''$. For each position $i$ of $t''$, we find whether $p''$ matches the substring of $t''$ that starts at $i$.

\paragraph{Combining the results} A substring of $t$ can be derived from $p$ iff there is a position $i$ such that two conditions hold: $p'$ matches a substring of $t'$ starting at the position $i$ and $p''$ matches a substring of $t''$ starting at position $i$. The first condition (coming from the Subset Matching instance) ensures that, if we match $p$ to a substring of $t$, the number of appropriate symbols  in $t$ is {\em at least} as large as in $p$. The second condition (coming from Wildcard Matching) ensures that, if $p$ contains a maximal substring consisting only of symbols $a$ (no $a^+$s) for some $a$, then the corresponding position in $t$ contains {\em exactly} the same number of symbols $a$. Notice that we do not need to satisfy this condition if $p$ starts or ends with a maximal substring consisting only of $a$s, which is why we start and end $p''$ with wildcards.
	\end{proof}
	
	\subsection{Algorithm for types ``$|{*}\circ$'' and ``$|{+}\circ$''}
\label{period}

\begin{theorem}
	Let $p$ be a regular expression of type ``$|{*}\circ$'' or ``$|{+}\circ$'' and $t$ be a text. In time $O(|p|+|t|)$ we can decide if $t$ can be derived from $p$.
\end{theorem}
\begin{proof}
	Let $t':=tt$ be the sequence $t$ repeated twice. By \cite{gusfield1997algorithms} (page 196), we can use the suffix tree and the lowest common ancestor data structures to preprocess the sequences $t'$ and $t$ in $O(|t|)$ time such that the following holds. For any $i=1,\ldots,|t|$, we can in $O(1)$ decide if the substring $t'_{i}t'_{i+1}\ldots t'_{i+|t|-1}$ of the sequence $t'$ of length $|t|$ is equal to the sequence $t$.
	
	Let $p'$ be an arbitrary sequence of non-zero length. We notice that the sequence $t$ can be derived from $[p']^*$ (or, equivalently, from $[p']^+$ if $|t|>0$) iff $p'$ is a prefix of $t$ and the substring $t'_{|p'|+1}t'_{|p'|+2}\ldots t'_{|p'|+|t|}$ is equal to $t$, and $|t|$ is divisible by $|p'|$. We can check the first condition in time $O(|p'|)$ and the second condition in time $O(1)$ assuming that we did the preprocessing step.
	
	Let $p=[p_1]^* \ | \ [p_2]^* \ | \ \ldots \ | \ [p_k]^*$ be the input pattern to the membership problem for some integer $k\geq 1$. We do the preprocessing step on $t'$ and $t$ which takes $O(|t|)$. For every $p_i$, $k\geq i \geq 1$ we decide in $O(|p_i|)$ time if $t$ can be derived from $[p_i]^*$. This yields the required runtime.
	
	The algorithm for the case when $p=[p_1]^+ \ | \ [p_2]^+ \ | \ \ldots \ | \ [p_k]^+$ is the same except we can't derive text $t$ if $|t|=0$.
\end{proof}

\subsection{Algorithms for types ``$|{\circ} +$'' and ``$*{\circ} +$''}
\label{run}

\begin{theorem}
	Let $p$ be a regular expression of type ``$|{\circ} +$'' or ``$*{\circ} +$'' and $t$ be a text. In time $O(|p|+|t|)$ we can decide if $t$ can be derived from $p$.
\end{theorem}
\begin{proof}
	We do the run-length encoding of $t$ defined as follows. Set $t':=t$ and initialize $r$ to be an empty ordered sequence of tuples. While $|t'|>0$, let $a$ be the first symbol of $t'$ and let $l>0$ be the largest integer such that $a^l$ (symbol $a$ repeated $l$ times) is a prefix of $t'$. Remove the prefix $a^l$ from $t'$ and add tuple $(a,l)$ at the end of the ordered sequence $r$, and repeat. This takes $O(|t|)$ time in total. Let $|r|$ be the number of tuples in $r$.
	
	\paragraph{Type ``$|{\circ} +$''} Let $p=p_1 \ | \ \ldots \ | \ p_k$ for some integer $k>0$, where each $p_i$ is a concatenation of $a$ and $a^+$ for various symbols $a$. We want to decide if there exists an integer $i=1, \ldots, k$ such that the text $t$ can be derived from the expression $p_i$. Fix an arbitrary $i=1, \ldots, k$. We do the run-length encoding of the expression $p_i$ and produce a sequence of tuples $r(p_i)$ defined as follows. Set $p_i':=p_i$ and $r(p_i)$ to be an empty sequence of tuples. While $|p_i'|>0$, choose the largest integer $l>0$ such that there exists a prefix of $p_i'$ of form $a a^+ a^+ a a \ldots$ (an arbitrary concatenation of $a$ and $a^+$) for some symbol $a$. Let $l'\geq 0$ be such that the prefix of $p_i'$ has $l'$ occurrences of $a$ and $l-l'$ occurrences of $a^+$. If $l'=l$, we add tuple $(a,=l)$ to the end of the sequence of tuples $r(p_i)$. Otherwise, if $l'<l$, we add tuple $(a,\geq l)$ to the end of the sequence $r(p_i)$. We delete the prefix of $p_i'$ and repeat (until $|p_i'|=0$). Let $|r(p_i)|$ be the number of tuples in $r(p_i)$. We can derive $t$ from the expression $p_i$ iff the following two conditions hold:
	\begin{itemize}
		\item $|r|=|r(p_i)|$.
		\item For all $j=1,\ldots,|r|$, if $r_j=(a,l)$ (the $j$-th tuple of $r$ is $(a,l)$), then $r(p_i)_j=(a,=l)$ or $r(p_i)_j=(a,\geq l')$ for some integer $l'\leq l$.
	\end{itemize}
	For every $i=1,\ldots,k$ this takes $O(|p_i|)$ time and the required upper bound on the runtime follows.
	
	\paragraph{Type ``$*{\circ} +$''} We have to consider the case when $p=[p']^*$ for some $p'$. We do the run-length encoding on the sequence $p'$ and get the sequence of tuples $r(p')$ (as described in the case for type ``$|\circ +$''). For every tuple $(a,=l)$ or $(a,\geq l)$, let $a$ be it's type. Consider two subcases.
	\paragraph{The first and the last tuple of $r(p')$ are of different types} We can derive the text $t$ from the expression $p$ iff the following two conditions
	\begin{itemize}
		\item $|r|$ is divisible by $|r(p')|$.
		\item Let $r'=(r(p'),r(p'),\ldots,r(p'))$, where $r(p')$ is repeated $|r|/|r(p')|$ time in the r.h.s.. For all $j=1,\ldots,|r|$, if $r_j=(a,l)$, then $r'_j=(a,=l)$ or $r'_j=(a,\geq l')$ for some integer $l'\leq l$.
	\end{itemize}
	
	\paragraph{The first and the last tuple of $r(p')$ are of the same type} If $|r(p')|=1$ we can check if $t$ can be derived from $p$ easily. If there is no integer $k\geq 1$ such that $|r|=k|r(p')|-(k-1)$, then $t$ can't be derived from $p$. Otherwise, let $r'$ be $r(p')$ except the last tuple. Let $r''$ be $r$ except we change the first tuple of $r$. Let $z$ be the first and the last tuple of $r$ merged (in the natural way). We replace the first tuple of $r$ by $z$ and get $r''$. We define $r'''=(r',r'',r'',r'',\ldots,r'')$, where $r''$ is repeated $k-1$ times. Furthermore, we add the last tuple of $r(p')$ at the end of $r'''$. The text $t$ can be derived from the expression $p$ iff for all $j=1,\ldots,|r|$, if $r_j=(a,l)$, then $r'''_j=(a,=l)$ or $r'''_j=(a,\geq l')$ for some integer $l'\leq l$.
	
	In the both subcases the runtime is upper bounded by $O(|p|+|t|)$.
\end{proof}
	
	\section*{Acknowledgments}
%\paragraph{Acknowledgments} 
We thank Ludwig Schmidt and the reviewers for  many helpful comments.
This work was supported by grants from the NSF, the MADALGO center, and the Simons Investigator award.

		\bibliographystyle{alpha}
		\bibliography{ref}
\end{document}